\documentclass{article}

\usepackage[utf8]{inputenc}

\def\showauthornotes{1}
\def\showkeys{0}
\def\showdraftbox{1}
\def\confversion{0}
\def\widemargin{1}
\def\ipadcompile{0}

\usepackage{xspace,enumerate}
\usepackage{amsmath,amssymb}
\usepackage{amsthm}
\ifnum\confversion=0
	\usepackage[toc,page]{appendix}
\fi
\usepackage[T1]{fontenc}
\usepackage[utf8]{inputenc}

\usepackage{thmtools}
\usepackage{thm-restate}
\usepackage{color,graphicx}
\usepackage{boxedminipage}
\usepackage{makecell}
\usepackage{tabularx,multirow}
\usepackage{csquotes}
\usepackage{mdframed}

\ifnum\widemargin=0
\usepackage[top=1in, bottom=1in, left=1in, right=1in]{geometry}
\else
\usepackage[top=1in, bottom=1in, left=1.25in, right=1.25in]{geometry}
\fi

\ifnum\showkeys=1
\usepackage[color]{showkeys}
\fi

\usepackage[dvipsnames]{xcolor}

\definecolor{darkred}{rgb}{0.5,0,0}
\definecolor{darkgreen}{rgb}{0,0.35,0}
\definecolor{darkblue}{rgb}{0,0,0.55}

\definecolor{outerblue}{RGB}{35,105,140}
\definecolor{concatyellow}{RGB}{255,190,20}
\definecolor{aelgreen}{RGB}{20,135,25}
\definecolor{randomorange}{RGB}{230,115,30}

\usepackage[pdfstartview=FitH,pdfpagemode=UseNone,colorlinks,linkcolor=NavyBlue,filecolor=blue,citecolor=OliveGreen,urlcolor=NavyBlue,pagebackref]{hyperref}

\usepackage{enumitem}
\usepackage{comment}
\usepackage[color=asparagus]{todonotes}

\usepackage{tikz}
\usetikzlibrary{
  arrows.meta,
  positioning,
  calc,
  decorations.pathreplacing
}
\usepackage[capitalise,nameinlink]{cleveref}

\usepackage{microtype}

\usepackage{mathtools,dsfont, bbm}

\usepackage{palatino}
\usepackage[scaled=.95]{helvet}

\ifnum\ipadcompile=1
\usepackage{eulervm}
\else
\usepackage{eulerpx}
\fi

\DeclareFontEncoding{LGR}{}{}
\DeclareSymbolFont{sfitgreek}{LGR}{cmss}{m}{it}
\SetSymbolFont{sfitgreek}{bold}{LGR}{cmss}{bx}{it}
\DeclareMathSymbol{\sfpi}{\mathord}{sfitgreek}{`p}

\usepackage{tcolorbox}
\tcbuselibrary{breakable}

\DeclareMathAlphabet{\mathpazocal}{OMS}{zplm}{m}{n}
\DeclareRobustCommand*{\mathcal}[1]{\mathpazocal{#1}}
\ifnum\showauthornotes=1
\newcommand{\Authornote}[3]{{\sf\small\color{#3}{[#1: #2]}}}
\newcommand{\Authorcomment}[2]{{\sf \small\color{gray}{[#1: #2]}}}
\newcommand{\Authorfnote}[2]{\footnote{\color{red}{#1: #2}}}
\else
\newcommand{\Authornote}[3]{}
\newcommand{\Authorcomment}[2]{}
\newcommand{\Authorfnote}[2]{}
\fi

\ifnum\showdraftbox=1

\else

\fi

\declaretheorem[numberwithin=section]{theorem}
\declaretheorem[sibling=theorem]{lemma}

\declaretheorem[sibling=theorem]{proposition}
\declaretheorem[sibling=theorem]{fact}
\declaretheorem[sibling=theorem]{corollary}

\theoremstyle{definition}
\declaretheorem[sibling=theorem]{definition}

\declaretheorem[sibling=theorem]{observation}

\newtheorem{algo}[theorem]{Algorithm}

\def\FullBox{\hbox{\vrule width 6pt height 6pt depth 0pt}}

\def\qed{\ifmmode\qquad\FullBox\else{\unskip\nobreak\hfil
\penalty50\hskip1em\null\nobreak\hfil\FullBox
\parfillskip=0pt\finalhyphendemerits=0\endgraf}\fi}

\def\qedsketch{\ifmmode\Box\else{\unskip\nobreak\hfil
\penalty50\hskip1em\null\nobreak\hfil$\Box$
\parfillskip=0pt\finalhyphendemerits=0\endgraf}\fi}

\DeclareMathOperator{\wt}{wt}

\newcommand{\inn}{\mathrm{inn}}
\newcommand{\gad}{0}
\newcommand{\out}{\mathrm{out}}
\newcommand{\base}{\mathrm{base}}
\newcommand{\ael}{\mathrm{AEL}}
\newcommand{\concat}{\mathrm{con}}
\newcommand{\sym}{\mathrm{sym}}

\newcommand{\capH}{\widetilde{H}}
\newcommand{\perpsym}{{\perp_{\sym}}}
\newcommand{\deltagv}{\delta_{\mathrm{GV}}}

\newcommand{\fold}[1]{\mathrm{Fold}\inparen{#1}}
\newcommand{\im}[1]{\mathrm{im}\inparen{#1}}

\newcommand{\stab}{\textup{Stab}}

\newcommand{\ketbra}[2]{\vert #1 \rangle \langle #2\vert}

\newcommand{\con}{\mathrm{con}}

\makeatletter
\newcommand{\stackalign}[1]{
	\vcenter{
		\Let@ \restore@math@cr \default@tag
		\baselineskip\fontdimen10 \scriptfont\tw@
		\advance\baselineskip\fontdimen12 \scriptfont\tw@
		\lineskip\thr@@\fontdimen8 \scriptfont\thr@@
		\lineskiplimit\lineskip
		\ialign{\hfil$\m@th\scriptstyle##$&$\m@th\scriptstyle{}##$\crcr
			#1\crcr
		}
	}
}

\let\latexcirc=\circ
\newcommand{\ccirc}{\mathbin{\mathchoice
  {\xcirc\scriptstyle}
  {\xcirc\scriptstyle}
  {\xcirc\scriptscriptstyle}
  {\xcirc\scriptscriptstyle}
}}
\newcommand{\xcirc}[1]{\vcenter{\hbox{$#1\latexcirc$}}}\let\circ\ccirc

\def\to{\rightarrow}
\def\eps{\varexpan}
\def\epsilon{\varepsilon}

\def\eps{\epsilon}

\def\phi{\varphi}

\def\implies{\Rightarrow}

\renewcommand{\bar}{\overline}

\newcommand{\ie}{i.e.,\xspace}

\newcommand{\etal}{et al.\xspace}

\newcommand{\mper}{\,.}

\newcommand{\F}{{\mathbb F}}

\newcommand{\cC}{\mathcal{C}}

\newcommand{\cQ}{\mathcal{Q}}

\newcommand{\cR}{\mathcal{R}}
\newcommand{\cN}{\mathcal{N}}

\DeclarePairedDelimiter\parens{\lparen}{\rparen}
\DeclarePairedDelimiter\abs{\lvert}{\rvert}
\DeclarePairedDelimiter\norm{\lVert}{\rVert}

\DeclarePairedDelimiter\braces{\lbrace}{\rbrace}
\DeclarePairedDelimiter\brackets{\lbrack}{\rbrack}
\DeclarePairedDelimiter\angles{\langle}{\rangle}
\DeclarePairedDelimiterXPP\lnorm[1]{}\lVert\rVert{_2}{#1}

\DeclareMathDelimiter{\given}
      {\mathbin}{symbols}{"6A}{largesymbols}{"0C}

\newcommand{\prob}{\mathsf{Pr}}
\newcommand{\Esymb}{\mathbb{E}}

\newcommand{\Psymb}{\mathrm{Pr}}

\DeclarePairedDelimiterXPP{\Prob}[1]
 {\prob}{\lparen}{\rparen}{}
 {\renewcommand{\given}{\;\delimsize\vert\nonscript\;\mathopen{}}#1}

\makeatletter
\def\Pr#1{%
    \ProbabilityRender{\Psymb}{#1}%
}

\def\Ex#1{%
    \ProbabilityRender{\Esymb}{#1}%
}

\def\condPE#1#2{%
	\@ifnextchar\bgroup
	{\ConditionalProbabilityRender{\widetilde{\Esymb}}{#1}{#2}}
	{\ProbabilityRender{\widetilde{\Esymb}}{#1 \given #2}}
}

\def\ConditionalProbabilityRender#1#2#3#4{
	\renderwithdist{#1}{#2}{#3 \given #4}	
}

\def\ProbabilityRender#1#2{%
  \@ifnextchar\bgroup%
  {\renderwithdist{#1}{#2}}
   {\singlervrender{#1}{#2}}
}
\def\singlervrender#1#2{%
   {\mathchoice
       {{#1}\brackets*{#2}}
       {{#1}[ #2 ]}
       {{#1}[ #2 ]}
       {{#1}[ #2 ]}
   }
}
\def\renderwithdist#1#2#3{%
   \@ifnextchar\bgroup
   {\superfancyrender{#1}{#2}{#3}}
   {\mathchoice
      {\underset{#2}{#1}\brackets*{#3}}
      {{#1}_{#2}[ #3 ]}
      {{#1}_{#2}[ #3 ]}
      {{#1}_{#2}[ #3 ]}
     }
}
\def\superfancyrender#1#2#3#4#5{
   \ensuremath{\mathchoice
      {\underset{#1}{{#1}}\left#4 #3 \right#5}
      {{#1}_{#2}#4 #3 #5}
      {{#1}_{#2}#4 #3 #5}
      {{#1}_{#2}#4 #3 #5}
   }
}
\makeatother

 \newcommand\SetSymbol[1][]{%
     \nonscript\:#1\vert
     \allowbreak
     \nonscript\:
     \mathopen{}}
  \DeclarePairedDelimiterX\Set[1]\{\}{%
     \renewcommand\given{\SetSymbol[\delimsize]}
     #1
}

\newcommand{\ip}[2]{\angles*{#1 , #2}}
\newcommand{\inparen}[1]{\parens*{#1}}             %
\newcommand{\inbraces}[1]{\braces*{#1}}

\DeclareMathOperator\supp{Supp}

\DeclareMathOperator{\Tr}{Tr}

\definecolor{asparagus}{rgb}{0.53, 0.66, 0.42}
\definecolor{cambridgeblue}{rgb}{0.64, 0.76, 0.68}
\definecolor{celadon}{rgb}{0.67, 0.88, 0.69}
\definecolor{charcoal}{rgb}{0.21, 0.27, 0.31}

\title{Random Quantum LDPC Codes Approaching the Gilbert--Varshamov Bound}
 \author{
 Tushant Mittal\thanks{{\tt Stanford University}. {\tt tushant@stanford.edu}. Supported by the NSF grants CCF-
 2143246 and CCF-2133154.}
 \and
 Shashank Srivastava\thanks{{\tt IIT Bombay}. {\tt shashanks@ttic.edu}.}
 \and
 Madhur Tulsiani\thanks{{\tt TTIC}. {\tt madhurt@ttic.edu}. Supported in part by NSF grant CCF-2326685.} 
 \and
 Mary Wootters\thanks{{\tt Stanford University}. {\tt marykw@stanford.edu}.  Supported in part by NSF grant CNS-2321489 and CCF-2231157.}
 }

\allowdisplaybreaks

\begin{document}

\date{}

\maketitle

\thispagestyle{empty}

\begin{abstract}
We show that for any $R,\eps >0$ and any prime $p\geq 2$, there exists an infinite family of $p$-ary quantum low-density parity-check (QLDPC) codes, rate $R$, checks of weight $O_\eps(1)$, and normalized distance at least $\deltagv(p,R)-\eps$. Here, $\deltagv(p,R)$ denotes the quantum Gilbert--Varshamov (GV) bound for $p$-ary stabilizer codes.

\medskip
In fact, we construct an ensemble of such QLDPC codes, such that a random code from this ensemble is close to the GV bound with high probability. Moreover, this ensemble matches the performance of random stabilizer codes on several quantum channels. Specifically, it approaches the quantum capacity of the erasure channel and the hashing bound for memoryless Pauli channels, including the depolarizing channel. 
\medskip

A significant challenge in working with QLDPC codes is that they are necessarily \textit{degenerate}, \ie contain many low-weight stabilizers. A key contribution of our work is a construction of QLDPC codes with quantitative control on their degeneracy.  These codes are obtained by combining known constructions of asymptotically good QLDPC codes with the expander-based distance amplification procedure of Alon, Edmonds, and Luby [FOCS'95].  

\medskip

Our random ensemble is constructed by starting with these low-degeneracy QLDPC codes near the quantum Singleton bound and concatenating each coordinate with a random inner code. This can be viewed as a quantum analogue of Thommesen's construction, and as an LDPC version of a result of Ouyang, with an appropriately designed outer code.

\end{abstract}

\newpage

\pagenumbering{roman}
\tableofcontents
\clearpage

\pagenumbering{arabic}
\setcounter{page}{1}

\section{Introduction}

Quantum stabilizer codes, first considered by Gottesman~\cite{Got97}, are quantum analogues of classical linear codes and are often the benchmark for several interesting properties of quantum codes. Just like a classical linear code can be defined via a parity-check matrix, $C = \ker(H)$, a stabilizer code can be defined using a single parity-check matrix $H$. To be precise, let us denote by $[[n,k,d]]_p$, a $p$-ary stabilizer code that encodes $k$ qudits ($p$-ary qubits) using $n$ qudits and has distance $d$, \ie can tolerate up to $\tfrac{d}{2}$ errors. Such a code can be defined using a matrix\footnote{The reader familiar with CSS codes should think of $H = \bigg(\begin{matrix}
	H_X&0\\
	0&H_Z 
\end{matrix}\bigg)$, and thus, $H\Lambda_n H^T = 0 \iff H_X H_Z^T = 0$. }:
\vspace{-3mm}
\begin{equation*}
H \in \F_p^{(n-k)\times 2n } \;\text{ such that }\; H\Lambda_n H^T = 0, \text{ where },\; \Lambda_n = \bigg(\begin{matrix}
	0&I_n\\
	-I_n&0 
\end{matrix}\bigg)\;. 	\vspace{-4pt}
\end{equation*}\
A family of $[[n,k,d]]_p$ codes is called "asymptotically good" if its rate $(R = k/n)$ and relative distance $(\delta = d/n)$ are constants independent of $n$.  Gottesman gave a randomized construction of stabilizer codes that are "asymptotically good". In fact, they achieve the best known tradeoff between rate and distance, which is captured by the quantum Gilbert--Varshamov (GV) bound.

\begin{theorem}[Quantum GV bound~\cite{Got97}]\label{thm:stabilizer_gv}
Fix any prime $p\geq 2$, and a rate $R \in (0,1)$. Let $H \in \F_p^{n(1-R)\times 2n}$ be a uniformly random matrix among all matrices satisfying $H\Lambda_n H^T = 0$. Then, with probability $1-\exp(-\Omega(n))$, the $[[n,Rn,d]]_{p}$ stabilizer code defined by $H$, achieves the GV bound, \ie
\[
d ~\geq~ (\deltagv(p,R) - o(1)) \cdot n
    ~:=~
    \left( H_{p^2}^{-1}\left(\frac{1-R}{2}\right)  - o(1) \right)\cdot n. \]
\end{theorem}

Quantum low-density parity check (QLDPC) codes are an important subclass of stabilizer codes which have the additional property that the check matrix $H$ is row-sparse, that is, has just $O(1)$ non-zero entries in each row. This LDPC property implies that errors can be detected by looking at just a few (qu)bits at a time, making operations very efficient, and thus is much desired for physical implementations. However, the LDPC condition requirement breaks the random argument above, as uniformly random sparse matrices do not satisfy $H\Lambda_n H^T = 0$.

This makes LDPC codes much more challenging to construct and analyze. There has been a long line of work~\cite{TZ14,EKZ20,KT21,HHO21} on constructing large-distance QLDPC codes starting from the work of Kitaev~\cite{Kit03}, culminating in the recent breakthrough result of Panteleev and Kalachev~\cite{PK22} (see also~\cite{LZ22,DHLV23,LZ23}), which showed that asymptotically good quantum QLDPC codes do exist. 
With this question settled, one is naturally led to ask if QLDPC codes can go all the way to the quantum GV bound? More generally, is LDPC a strong restriction~or,  
\begin{center}
\textit{Can the performance of quantum stabilizer codes always be matched by quantum LDPC codes?
}\end{center}

In fact, the above hides a more basic question about the distinction between explicit and probabilistic constructions. 
For classical codes, Gallager~\cite{Gallager62}  constructed an ensemble of LDPC codes and proved that a random (binary) code from this family attains the GV bound. Since then, his ensemble has shed light on various properties of classical LDPC codes~\cite{MRRSW20,GLM+20}. 
On one hand, the most likely avenue for answering the above question about QLDPC codes is also via probabilistic construction since we do not know of explicit codes matching the GV bound even without the LDPC property (or even classically!). 

On the other hand, even proving the \emph{existence} of asymptotically good QLDPC required explicit constructions, as we do not know of any random ensembles for such codes.
The known random models for classical LDPC codes do not extend to the setting of quantum codes, due to the $H\Lambda_n H^T = 0$ constraint, and the lack of a QLDPC ensemble prevents us from understanding the capabilities of QLDPC codes. 
Moreover, random ensembles achieving the GV would also \emph{imply} the existence of good QLDPC codes, which was challenging to establish.
We show that this implication is in some sense the ``only" barrier to defining random models, and construct an ensemble of QLDPC codes  \emph{on top of} existing explicit constructions. 
As a result, we obtain the first ensemble of QLDPC codes such that the resulting random codes achieve the GV bound, and also match random stabilizer codes for many other applications.

\begin{theorem}[Main Theorem, Informal]\label{thm:eta_stabilizer_gv}
Fix any prime $p \geq 2$, a rate $R \in (0,1)$, and a threshold $\eps \in (0,1)$. Then, for each large enough $n$, there exists an ensemble of QLDPC $[[n, Rn, d]]_p$ stabilizer codes such that each code has checks of weight $O_{\eps, R}(1)$, and except with probability $2^{-\Omega_\eps(n)}$ , the code approaches the GV-bound, \ie  has distance $d \geq (\deltagv(p,R) -\eps) 
	\cdot n$. 
\end{theorem}

Our random ensemble starts with an asymptotically good LDPC base code, amplifies its distance via the technique of Alon, Edmonds, and Luby (AEL), and finally introduces randomness in it locally by concatenating each symbol with an independent random code. 
At a high level, the ensemble inherits the LDPC property from the global outer code, and simultaneously, the local randomness suffices for it to behave like a truly random stabilizer code. 
Such ensembles have long been known for classical codes~\cite{T83} and have also been considered for quantum Reed-Solomon codes~\cite{Ouy14}. 
However, working with quantum LDPC codes presents some unique challenges discussed later.

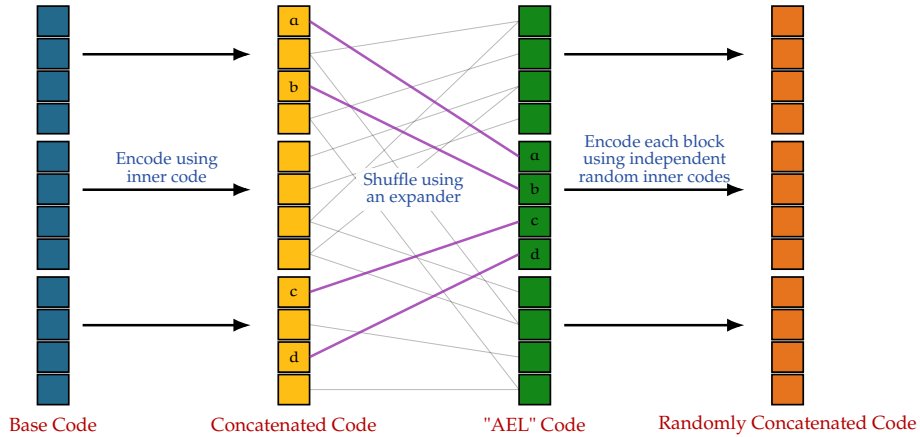
\begin{figure}[!h]
\centering

\definecolor{outerblue}{RGB}{35,105,140}
\definecolor{concatyellow}{RGB}{255,190,20}
\definecolor{aelgreen}{RGB}{20,135,25}
\definecolor{randomorange}{RGB}{230,115,30}
\definecolor{highlightpurple}{RGB}{150,45,170}
\definecolor{processtext}{RGB}{40,80,160}

\resizebox{0.8\textwidth}{!}{%
\begin{tikzpicture}[
    >=Latex,
    thick,
    every node/.style={font=\small},
    codearrow/.style={->, very thick},
    expedge/.style={thin, opacity=.28},
    highedge/.style={very thick, highlightpurple, opacity=.8},
    processlabel/.style={
        text=processtext,
        font=\small,
        align=center,
        scale=0.9
    },
]

\def\xouter{0}
\def\xconcat{4.0}
\def\xael{8.0}
\def\xrandom{12.2}

\def\cellw{0.52}
\def\cellh{0.50}
\def\gap{0.04}

\def\ytop{4.4}
\def\ymid{2.15}
\def\ybot{-0.10}

\newcommand{\drawblock}[3]{%
    \foreach \j in {1,...,4} {%
        \filldraw[fill=#3,draw=black]
            (#1,{#2-(\j-1)*(\cellh+\gap)})
            rectangle ++(\cellw,-\cellh);
    }%
}

\drawblock{\xouter}{\ytop}{outerblue}
\drawblock{\xouter}{\ymid}{outerblue}
\drawblock{\xouter}{\ybot}{outerblue}

\drawblock{\xconcat}{\ytop}{concatyellow}
\drawblock{\xconcat}{\ymid}{concatyellow}
\drawblock{\xconcat}{\ybot}{concatyellow}

\draw[codearrow]
    ({\xouter+0.75},{\ytop-0.80})
    --
    ({\xconcat-0.45},{\ytop-0.80});

\draw[codearrow]
    ({\xouter+0.75},{\ymid-0.80})
    --
    node[above,processlabel]
        {Encode using\\[-1mm] inner code}
    ({\xconcat-0.45},{\ymid-0.80});

\draw[codearrow]
    ({\xouter+0.75},{\ybot-0.80})
    --
    ({\xconcat-0.45},{\ybot-0.80});

\drawblock{\xael}{\ytop}{aelgreen}
\drawblock{\xael}{\ymid}{aelgreen}
\drawblock{\xael}{\ybot}{aelgreen}

\draw[expedge]
    ({\xconcat+\cellw},3.61) -- (\xael,4.15);

\draw[expedge]
    ({\xconcat+\cellw},2.53) -- (\xael,3.61);

\draw[expedge]
    ({\xconcat+\cellw},1.90) -- (\xael,3.07);

\draw[expedge]
    ({\xconcat+\cellw},1.36) -- (\xael,2.53);

\draw[expedge]
    ({\xconcat+\cellw},0.82) -- (\xael,-0.35);

\draw[expedge]
    ({\xconcat+\cellw},0.28) -- (\xael,-0.89);

\draw[expedge]
    ({\xconcat+\cellw},-0.89) -- (\xael,-1.43);

\draw[expedge]
    ({\xconcat+\cellw},-1.97) -- (\xael,-1.97);

\draw[expedge]
    ({\xconcat+\cellw},3.61) -- (\xael,-0.89);

\draw[expedge]
    ({\xconcat+\cellw},2.53) -- (\xael,-1.97);

\draw[expedge]
    ({\xconcat+\cellw},0.82) -- (\xael,4.15);

\draw[expedge]
    ({\xconcat+\cellw},0.28) -- (\xael,3.07);

\draw[highedge]
    ({\xconcat+\cellw},4.15)
    --
    (\xael,1.90);

\draw[highedge]
    ({\xconcat+\cellw},3.07)
    --
    (\xael,1.36);

\draw[highedge]
    ({\xconcat+\cellw},-0.35)
    --
    (\xael,0.82);

\draw[highedge]
    ({\xconcat+\cellw},-1.43)
    --
    (\xael,0.28);

\node[font=\scriptsize]
    at ({\xconcat+\cellw/2},4.15)
    {$a$};

\node[font=\scriptsize]
    at ({\xconcat+\cellw/2},3.07)
    {$b$};

\node[font=\scriptsize]
    at ({\xconcat+\cellw/2},-0.35)
    {$c$};

\node[font=\scriptsize]
    at ({\xconcat+\cellw/2},-1.43)
    {$d$};

\node[font=\scriptsize]
    at ({\xael+\cellw/2},1.90)
    {$a$};

\node[font=\scriptsize]
    at ({\xael+\cellw/2},1.36)
    {$b$};

\node[font=\scriptsize]
    at ({\xael+\cellw/2},0.82)
    {$c$};

\node[font=\scriptsize]
    at ({\xael+\cellw/2},0.28)
    {$d$};

\node[
    processlabel,
    fill=white,
    fill opacity=.90,
    text opacity=1,
    inner xsep=4pt,
    inner ysep=2pt
]
    at (6.25,1.35)
    {Shuffle using\\[-1mm] an expander};

\drawblock{\xrandom}{\ytop}{randomorange}
\drawblock{\xrandom}{\ymid}{randomorange}
\drawblock{\xrandom}{\ybot}{randomorange}

\draw[codearrow]
    ({\xael+0.75},{\ytop-0.80})
    --
    ({\xrandom-0.45},{\ytop-0.80});

\draw[codearrow]
    ({\xael+0.75},{\ymid-0.80})
    --
    node[above,processlabel]
        {Encode each block\\[-1mm]
         using independent\\[-1mm]
         random inner codes}
    ({\xrandom-0.45},{\ymid-0.80});

\draw[codearrow]
    ({\xael+0.75},{\ybot-0.80})
    --
    ({\xrandom-0.45},{\ybot-0.80});

\node[text=red!70!black]
    at ({\xouter+\cellw/2},-2.55)
    {Base Code};

\node[text=red!70!black]
    at ({\xconcat+\cellw/2},-2.55)
    {Concatenated Code};

\node[text=red!70!black]
    at ({\xael+\cellw/2},-2.55)
    {"AEL" Code};

\node[text=red!70!black,align=center]
    at ({\xrandom+\cellw/2},-2.55)
    {Randomly Concatenated Code};

\end{tikzpicture}%
}

\vspace{-5 pt}
\caption{A schema of our random model.}
\label{fig:model}

\end{figure}
\vspace{-5 pt}

Going beyond the GV-bound, we show that our random ensemble of quantum LDPC codes also matches the performance of random stabilizer codes on multiple quantum channels of interest. Specifically, we consider the (i) \textit{quantum erasure channel}, which erases each qudit with probability $\kappa$, replacing it with a special symbol $\bot$ (so that the erasure locations are known), and the family of (ii) \textit{memoryless Pauli channels}, specified by a probability distribution $\nu$ over 1-qudit Pauli operators, which independently applies a random Pauli operator (drawn from $\nu$) to each of the $n$ qudits. 

It is known that a random stabilizer can achieve capacity on the erasure channel, \ie transmit information reliably with rate $R = 1-2\kappa$, whereas for the memoryless Pauli channel, the ``hashing bound" shows that it can achieve a rate of $1 - \frac{1}{\log_2 p} \cdot H(\nu)$, where $H(\nu)$ is the binary entropy of the distribution $\nu$. 
For each channel $\mathcal{Q}$ above, let $R(\mathcal{Q})$ denote the rates above, achieved by random stabilizer codes.
We prove that our QLDPC ensemble can match this up to an arbitrary small $\eps$-loss.

\begin{theorem}[Performance on Channels]\label{thm:eta_channel}
Fix any prime $p \geq 2$. Let $\mathcal{Q}$ be the erasure channel or any memoryless Pauli channel, and let $R(\mathcal{Q})$ be the rate bounds above achieved by a random $p$-ary stabilizer code for the channel $\mathcal{Q}$. Then, for any $\eps > 0$, there exists an ensemble of $p$-ary QLDPC stabilizer codes with checks of weight $O_\eps(1)$ such that random codes from this ensemble achieve a rate of $R(\mathcal{Q}) -\eps$ over the channel $\mathcal{Q}$.  
\end{theorem}

Both the GV bound and these channel performances of the random stabilizer codes can be deduced from a very useful property that random stabilizer codes satisfy which is that for any fixed error pattern $e$, the probability that a random rate $R$ stabilizer code cannot detect $e$ is tiny:  
\begin{align*}
   	 \Pr{\text{Random Stabilizer Code}}{\;e \text{ causes an undetectable error} \,}
    ~&\leq~
    p^{-N{(1-R)}}\;.
\end{align*}
This is the random property that our ensemble (see \cref{thm:final_rate1})  mimics, 
\begin{align*}
   	 \Pr{\text{Our QLDPC Ensemble}}{\;e \text{ causes an undetectable error} \,}
    ~&\leq~
    p^{-N{(1-R-\eps)}}\;.
\end{align*}

In \cref{sec:tech} below, we discuss how we establish this statement for our ensemble of random LDPC codes, as well as why it is useful for establishing \cref{thm:eta_stabilizer_gv} and \cref{thm:eta_channel}.

\paragraph{Other random models.} 
Recent works have proposed several ensembles for quantum codes. However, the resulting codes are either non-LDPC or are not known to achieve good rate-distance tradeoffs. For example, Okada and Kasai~\cite{OK26} consider a random model for quantum LDPC codes via random swaps which seem to have good distance in experiments, but no formal distance guarantees are known. 
Tremblay \etal ~\cite{TDCK23} also consider a different random model for LDPC stabilizer codes, which they evaluate experimentally on the random erasure channel.
Memoryless Pauli channels were also considered in the work of Komoto and Kasai~\cite{KK25}, which studies certain ensembles of quantum LDPC codes for these channels via numerical experiments.

Another ensemble of random codes which achieve the (CSS) GV bound is considered by Kasai~\cite{Kas26}. However, these are not LDPC codes but \emph{subcodes} (puncturings) of LDPC codes\footnote{Puncturings of LDPC codes may not be LDPC. In fact, \emph{all} classical linear codes are puncturings of the Hadamard code, which is LDPC. However, the construction of Kasai~\cite{Kas26} also gives other structural properties for the resulting ensemble.}
Motivated by the complexity of encoding, Brown and Fawzi~\cite{BF13} and Anand \etal~\cite{AGHS26} consider a model of random codes defined by random shallow Clifford circuits. 
They also show that the resulting codes achieve the quantum GV bound, though they are not necessarily LDPC.

\subsection{Techniques}\label{sec:tech}
In this work, we use the linear-algebraic language of stabilizer spaces for quantum codes. We introduce the required notation here for our discussion; complete definitions can be found in~\cref{sec:prelims}. 

As mentioned earlier, a $p$-ary stabilizer code encoding $k$ qudits using $n$ qudits is defined by a matrix $H \in F_p^{(n-k) \times 2n}$ $H\Lambda_n H^T = 0$. The row-space $L$ of the matrix $H$ is referred to as the \emph{stabilizer space} for the code\footnote{This can be thought of as the space of all Pauli operators that (pointwise) ``stabilize" the code space. The extra factors of 2 arise from the fact that the $p$-ary Pauli group has size $p^2$ and thus the stabilizer space is really a subspace of $(\F_{p^2})^n$.}. The condition that $H\Lambda_n H^T = 0$ corresponds to the requirement that $L$ is ``isotropic" \ie $L \leq L^{\perpsym}$ where $L^{\perpsym}$ is the dual subspace of $L$ according to the symplectic product.

The quantum code given by an isotropic subspace $L \leq \F_p^{2n}$ corresponds to the space of cosets $L^{\perpsym}/L$.
We will denote this code as $\cC(L)$ or simply $\cC$ when $L$ is clear from context.
We will use $d$ to denote the distance of a code and $\delta$ to denote its normalized distance $d/n$. 
Also, when suppressing the distance of a code, we refer to it simply as an $[[n,k]]_p$ quantum code.

The above definitions for $p$-ary quantum codes can also be extended to codes over $\F_q$ when $q=p^m$ is a prime power, but the analogue of the matrix $\Lambda_n$ requires a bit more care. 
However, we can still work with the notion of isotropic spaces $L$ (now over $\F_p^{2mn}$) as in the discussion below. 

\paragraph{Randomly concatenated codes.}
The construction we consider is a variant of the well-known construction of Thommesen~\cite{T83} for classical codes.  
One starts with an "outer" code $\cC_{\out} \subseteq \F_q^n$ over a large alphabet $q = p^m$, and encodes the symbol in each position $i \in [n]$, using a random and independent inner code $\cC_i: \F_q \cong \F_p^m \to \F_p^{n_{\inn}}$. 
The resulting (random) code $\cC_{\con}$ is known to achieve the (classical) GV bound when $\cC_{\out}$ is the Reed-Solomon code as in the analysis of Thommesen~\cite{T83}, or indeed any code close to the Singleton bound with $\delta_{\out} \approx 1 - R_{\out}$~\cite{GI04}. %
Moreover, the random codes can also be shown to admit efficient decoding~\cite{GI04} when $\cC_{\out}$ admits such decoding, and are LDPC when $\cC_{\out}$ is LDPC and $m$ is constant. Building on~\cite{GI04}, \cite{GKORZS18} showed that these random codes also enable local algorithms for error-detection and error-correction.

The quantum analogue of this construction was first considered by Ouyang~\cite{Ouy14}, who proved that when $\cC_{\out}$ is the quantum Reed-Solomon code, the resulting random codes achieve the quantum GV bound with high probability. 
\begin{theorem}[Randomly concatenated Reed-Solomon codes approach the GV bound~\cite{Ouy14}]\label{thm:ouyang_gv}
Fix rates $r,R \in (0,1)$. Let $L \leq \F_p^{2mn}$ be a rate $R$ quantum Reed-Solomon code over alphabet $p^m$. Set $n_\inn= m/r$, and let $L_\inn = \oplus_{i=1}^n L_i$ be the subspace where each $L_i \leq \F_p^{2n_{\inn}}$ is a randomly chosen isotropic subspace defining an $[[n_{\inn},m]]_{p}$ stabilizer code. Then, with probability $1-\exp(-\Omega(n))$, the distance of the  concatenated code $L_\con$ satisfies:
\[
d(L_\con) ~\geq~ (\deltagv(p,rR)- o_n(1))\cdot nn_\inn \mper 
\]
\end{theorem}
Note that Reed-Solomon codes are not LDPC and also require a large alphabet size $q = p^m \geq n$, and thus the above random ensemble does not define random LDPC codes. 
Given that concatenation preserves LDPC structure when $m$ is small, our approach is to replace the outer code by an LDPC quantum stabilizer code over a constant-sized alphabet, approaching the quantum Singleton bound. However, implementing the above approach with quantum LDPC codes as outer codes requires addressing two important issues, which we describe next. 

\paragraph{Choice of outer code.}
The first issue is to actually \emph{find} such a construction of quantum LDPC codes with a constant-sized alphabet, approaching the Singleton bound, which can serve as the outer code for our construction! 
Note that while the GV bound depends on both rate and alphabet size of the code, the Singleton bound only depends on the rate. 
In particular, it is easier to obtain codes approaching the Singleton bound using a suitably large alphabet size, while also maintaining the LDPC property.

One way to construct such a code is via the beautiful expander-based distance amplification procedure of Alon, Edmonds, and Luby (AEL)~\cite{AEL95}. This procedure amplifies a ``base" code (with high rate) and yields one with parameters approaching the Singleton bound. The quantum analogue of the classical AEL construction was analyzed by Bergamaschi, Golowich, and Gunn~\cite{BGG24} and is known to yield codes approaching the quantum Singleton bound (also see ~\cite{BJMST24} for a slightly simpler analysis). Importantly for us, the AEL procedure preserves the LDPC property for the starting base code, and has a constant-sized alphabet for the resulting code.

Thus, taking the base code to be one of the recent constructions of good quantum LDPC codes~\cite{PK22,LZ22}, we obtain an outer quantum LDPC code over a constant-sized alphabet whose parameters approach the quantum Singleton bound. Our random model uses this as the outer code in the Thommesen construction. Note that our random model is therefore built \emph{on top of} existing good quantum LDPC codes, after applying the AEL construction. In particular, it does not require a new proof of the existence of asymptotically good quantum LDPC codes.

\paragraph{Analyzing randomly concatenated quantum LDPC codes.}
The second issue, which arises from the very nature of quantum LDPC codes, is much more fundamental and is the main challenge in implementing the above approach.

At a high level, any result about a probabilistic model of random codes requires a union bound, which involves counting the "bad events" that must be avoided by the random code. 
For $y \in \cC_{\out} \leq \F_{q^2}^n$ (with $q = p^m$), let $\phi_i: \F_p^{2m} \to \F_p^{2n_{\inn}}$ be the independent linear maps\footnote{These linear maps will additionally need to be symplectic maps as described in \cref{sec:prelims}, but we ignore this for the discussion here.} corresponding to the inner codes.
The resulting code having distance $\delta_{\con}$ is equivalent to bounding the probability
\[
\Pr{\exists ~y \in \cC_{\out}\setminus\{0\} ~~\text{such that}~~ \wt(\phi_1(y_1), \ldots, \phi_n(y_n)) < \delta_{\con} \cdot n} \mper
\]
For codewords $y \in \cC_{\out}$ with support $S = \inbraces{i ~|~ y_i \neq 0}$, such an event only depends on the randomness of $\inbraces{\phi_i}_{i \in S}$ since $\phi_j(y_j) = \phi_j(0) = 0$ for all $j \notin S$.
When $\cC_\out$ is a Reed-Solomon code (classical or quantum) of distance $d$, the known proofs make use of the following fact: any two codewords $y,y'$ with support $S$, which agree in $\abs{S}-d+1$ positions, must agree everywhere (otherwise $y-y'$ would be a codeword with weight less than $d$, which contradicts the definition of distance).
This gives a count of $(p^{2m})^{\abs{S}-d+1}$ for codewords with support $S$, which is then used for the above union bound. 
This is the only place where the near-optimal distance of the outer code is used, but this counting relying only on the distance \emph{fails} for quantum LDPC codes.

The issue is that for a stabilizer quantum code with space of stabilizers $L \leq \F_p^{2mn}$, the codewords actually correspond to the \emph{cosets} in $L^{\perpsym}/L$.
The distance of the code is the minimum weight of a vector in $L^{\perpsym} \setminus L$, and not simply the minimum over all non-zero vectors. 
Thus, the above argument only gives that when $y$ and $y'$ with support $S$ agree in $\abs{S}-d+1$ positions, their difference $y-y'$ \emph{cannot be in $L^{\perpsym}\setminus L$} and must lie in $L$.
Quantum Reed-Solomon codes are ``pure" codes \ie the distance guarantee is actually on all of $L^{\perpsym} \setminus \{0\}$ (equivalently, $L$ also does not have low-weight vectors) and the above issue does not arise.

However, for any LDPC code, the space of stabilizers (parity checks) $L$ is actually \emph{generated} by the constant-weight parity checks, and will always contain low-weight vectors!
This can potentially increase the count above by a factor of $\abs{L}$, which does not yield any nontrivial bounds.
We thus cannot simply rely on the distance of the outer quantum LDPC code, and need to study their structure to understand that if their stabilizer space $L$ must unavoidably contain 
low-weight vectors (stabilizers), how many can there be?

\paragraph{Degeneracy.}
We now come to the key property we need to understand for quantum LDPC (outer) codes. 
For $\eta \in [0,1]$, we say that a code given by a stabilizer space $L$ and with minimum distance at least $d$ is \emph{$\eta$-degenerate} if
\[
\left|\{\ell\in L:\wt(\ell) < d\}\right| ~\leq~ |L|^\eta.
\]
Note that every code is 1-degenerate and a pure code like Reed-Solomon is $0$-degenerate. 
Also, if we have a code which is $\eta$-degenerate, the overhead in the counting argument above is $\abs{L}^{\eta}$, which can be handled if $\eta$ is very small.
The key result that drives our approach is showing that for any $R, \eps, \eta > 0$ the AEL construction can be used to design quantum LDPC codes with rate $R$ and distance at least $d = \frac{(1-R-\eps)}{2} \cdot n$, which are also $\eta$-degenerate. 
\begin{restatable}[Low-Degeneracy QLDPC Codes]{theorem}{AELCode}\label{thm:ael_eta_good}
Fix a prime $p$, a rate $R\in(0,1)$, and $\eta,\eps \in (0,1)$. Then there exists $m, n_0$ depending on $\eta, R$, such that for every $n\geq n_0$, there exists a $[[n,Rn,d]]_{p^m}$ stabilizer code $\mathcal{F}_{n}$ $\eta$-degenerate, with checks of weight at most $1/(\eta\cdot\eps)^{O(1)}$, and has  distance $d(\mathcal{F}_n) \geq \frac{1-R-\eps}{2} \cdot n$.
\end{restatable}

Hence, the AEL construction can be seen as approximating Reed-Solomon codes in multiple ways while also yielding additional properties, such as a constant-sized alphabet and an LDPC structure. 
Classical results already show that instead of being \emph{at} the Singleton bound with distance $\delta = 1-R$, AEL codes can have distance at least $\delta = 1 - R - \eps$ for any $\eps > 0$ (and this becomes $(1-R-\eps)/2$ for quantum codes since the quantum Singleton bound is $(1-R)/2$.)
We show that they also give an approximate version of the 0-degeneracy of Reed-Solomon codes, relaxing it to $\eta$-degeneracy. 

\paragraph{The AEL construction.}
We now briefly describe the AEL procedure and how it can be used to construct low-degeneracy LDPC codes. 
As depicted in \cref{fig:model}, the AEL construction can \emph{also} be viewed as code concatenation, but with qudits shuffled using an expander graph. 
For this, we consider a $\Delta$-regular bipartite graph $G=(U,V,E)$ with $\abs{U}=\abs{V}=n$. 
For $\lambda \in (0,1)$, we say that the graph is a $\lambda$-expander if the second singular value of the normalized biadjacency matrix is at most $\lambda$.
For such graphs, the expander mixing lemma can be used to estimate the number of edges between any $S \subseteq U$ and $T \subseteq V$ as $\abs{E(S,T) - \frac{\Delta}{n} \cdot \abs{S}\abs{T}} \leq \lambda \cdot \Delta \cdot n$.

One starts with a base code $\cC_{\base} \subseteq (\F_{p}^{2m})^n$, and encodes each symbol of a codeword $y \in \cC_{\base}$ with the \emph{same} inner code $\cC_{\gad}: \F_{p}^{2m} \to \F_{p}^{2\Delta}$.
We associate positions in the base code with the left vertex set $U$, and the positions in the concatenated code with the edge set $E$, since each left vertex has exactly $\Delta$ edges. 
Finally, we ``fold" the symbols on the edges incident to each right vertex $v \in V$ into a single symbol over the larger alphabet $\F_p^{2\Delta}$, giving the resulting code $\cC_{\ael} \subseteq (\F_p^{2\Delta})^n$.
Note that if $\cC_{\base}$ has checks of weight at most $w$, then $\cC_{\ael}$ has checks of weight at most $w\cdot\Delta$ and thus preserves the LDPC property.

For a small enough parameter $\gamma > 0$, we choose the rate of the base code $R_{\base} = 1-\gamma$ and the let the relative distance $\delta_{\base}$ be the appropriate constant given by an existing construction of good QLDPC codes~\cite{PK22,LZ22,LZ23}. 
Let the inner code be a (constant sized) code achieving the Singleton bound, say with rate $R_{\gad}$ and relative distance $\delta_{\gad} = (1-R_{\gad})/2$. We also assume the inner code is pure and can take it to be the quantum Reed-Solomon code for example (note that we only need the base code to be LDPC, and not the inner code).
The analysis of the AEL procedure for quantum codes (see \cref{sec:AEL}) shows that the resulting code has rate and relative distance,
\[
R_{\ael} ~=~ (1-\gamma) \cdot R_{\gad}
\qquad \text{and} \qquad
\delta_{\ael} ~\geq~ \delta_{\gad} - \frac{O(\lambda)}{\delta_{\base}} \mper
\]
Thus, by choosing a graph with $\lambda$ small enough, we can get a code with parameters approaching the Singleton bound with $\delta_{\ael} \geq (1-R_{\ael} - 4\gamma)/2$.

\paragraph{Low-degeneracy codes via AEL.}
To analyze the degeneracy, we need to analyze the low-weight vectors in the stabilizer space $L_{\ael}$ for the code $\cC_{\ael}$. 
Let the stabilizer spaces of the base and inner codes, and the encoding map applied to all $n$ coordinates, be denoted as follows: 
\[
\begin{gathered}
L_{\base} \subseteq \F_p^{2mn},
\qquad L_{\gad} \subseteq \F_p^{2\Delta},
\qquad \Phi := \phi^{\oplus n}, \qquad L_{\ael} ~=~ \fold{\Phi(L_{\base}) ~\oplus~ L_{\gad}^{\oplus n}}\;.
\end{gathered}
\]
Here $\fold{\cdot}$ is just a linear operator grouping the coordinates corresponding to the edges incident at each right vertex, and the encoding map is chosen so that the direct sum above is well defined. Moreover, since the inner code $\cC_{\gad}$ is pure,
\[
\im{\phi} \cap L_{\gad} ~=~ \{0\},
\qquad \im{\phi} \oplus L_{\gad} ~=~ L_{\gad}^{\perpsym},
\qquad \wt(v) ~\geq~ \delta_{\gad}\Delta
\quad\text{for all }0\neq v\in L_{\gad}^{\perpsym}.
\]

This special structure of the AEL code helps us analyze its low-weight vectors. Let $z \in L_\ael$ be any vector which can be expressed as $z := \fold{\Phi(x)+y}$, for $x \in L_{\base}$ and $y \in L_{\gad}^{\oplus n}$, and define: 
\begin{align*}
S ~&:=~ \supp_U(\Phi(x)+y)
~=~ \braces[\big]{u \in U \mid \phi_u(x_u) + y_u \neq 0},
&
T ~&:=~ \supp_V(\Phi(x)+y).
\end{align*}
The goal is to bound the number of such vectors $z \in L_{\ael}$ for which $\abs{T}$ is small. We identify the $\Delta n$ coordinates before folding with the edges of the graph. Under this identification, the right endpoint of every edge with a nonzero label must lie in $T$.
Each $u \in S$ has at least $\delta_{\gad}\Delta$ adjacent edges with nonzero labels. Combining this with the expander mixing lemma gives
\[
\begin{gathered}
\delta_{\gad}\Delta\abs{S}
~\leq~ E(S,T)
~\leq~ \frac{\Delta}{n}\abs{S}\abs{T}+\lambda\Delta n
\quad\implies\quad
\abs{S} ~\leq~ \frac{\lambda}{\delta_{\gad}-\abs{T}/n}\,n,\\[-2pt]
\abs{T} \leq (\delta_{\gad}-\gamma)n,
\quad \lambda\leq\gamma^2
\quad\implies\quad
\abs{S} ~\leq~ \frac{\lambda}{\gamma}n \leq \gamma n.
\end{gathered}
\]

Note that for $u\notin S$, we must have $\phi(x_u) = 0$ and $y_u = 0$ since they both lie in complementary spaces.
For each $u \in S$, we have at most $p^{2\Delta}$ choices for $\phi(x_u)$ (and hence for $x_u$) and at most $p^{2\Delta}$ choices for $y_u$. For $\gamma\leq1/2$, summing over all choices of $S$, and using the size of the stabilizer space, gives
\[
\begin{gathered}
\sum_{s\leq\gamma n}\binom ns p^{4\Delta s}
~\leq~ 2^{nH_2(\gamma)}p^{4\Delta\gamma n},
\qquad
\abs{L_{\ael}} ~=~ p^{\Delta n(1-R_{\ael})},\\[-2pt]
H_2(\gamma)\log_p2+4\Delta\gamma
~\leq~ \eta\Delta(1-R_{\ael}).
\end{gathered}
\]
The last inequality, which holds for sufficiently small $\gamma>0$, implies that the code is $\eta$-degenerate.

\paragraph{Using low-degeneracy to obtain good random ensembles.} Recall that our random model applies a random concatenation to a degenerate code. Our goal is to bound the probability that any fixed error $e$ is undetectable, \ie $e \in L^\perpsym\setminus L$. For this to occur, one of the vectors from the outer code must be mapped to $e$ by the random concatenation.  Degeneracy allows us to bound the number of outer code vectors that can map to any such $e$, while the randomness of the concatenation controls the probability of such a collision. Together, these give essentially the same bound for a fixed error as a random stabilizer code, which is the key input for our main results.

\section{Preliminaries}\label{sec:prelims}

This section defines the linear-algebraic model for the stabilizer codes we will work with in this paper. We then define the concatenation of stabilizer codes and our random model.

\subsection{Linear-Algebraic Formulation}

We use the standard additive symplectic description of stabilizer codes; see
\cite{La20} for background. Let $p$ be prime, let $q=p^m$, and let
$\Tr_{q/p}(a):=\sum_{i=0}^{m-1}a^{p^i}$ denote the field trace from $\F_q$ to
$\F_p$. Up to a scalar phase, the label $(a,b)\in\F_q^n\oplus\F_q^n$ specifies
the Pauli operator $P((a,b)) :=  X(a)Z(b)$, whose action on a computational-basis state is
\[X(a)Z(b)\,\lvert x\rangle ~=~\omega_p^{\Tr_{q/p}(\langle b,x\rangle)}\,\lvert x+a\rangle, \;
\text{ for } x\in\F_q^n,\quad \text{ where }\;\omega_p~=~e^{2\pi i/p}. \] The commutativity of two such operators is captured by the trace-symplectic bilinear form:
\[
    [(a,b),(a',b')]
    ~:=~
    \Tr_{q/p}\left(
        \langle a,b'\rangle-\langle b,a'\rangle
    \right),
    \qquad
    \langle a,b\rangle~:=~\sum_{i~=~1}^n a_i b_i.
\]
The corresponding Pauli operators commute if and only if
$[(a,b),(a',b')]=0$. We suppress scalar phases because they affect neither
commutation nor any code parameter. We can express the operator's label over the base prime field as follows. Fix a
basis $\alpha_1,\ldots,\alpha_m$ of $\F_q$ over $\F_p$ and its trace-dual basis
$\beta_1,\ldots,\beta_m$, so that
$\Tr_{q/p}(\alpha_i\beta_j)=\delta_{ij}$. For every coordinate $j\in[n]$, write
\[
    a_j~=~\sum_{\ell~=~1}^m x_{j,\ell}\alpha_\ell\;,
    \qquad
    b_j~=~\sum_{\ell~=~1}^m z_{j,\ell}\beta_\ell\;.
\]
This identifies $(a,b)$ with $(x,z)\in\F_p^{mn}\oplus\F_p^{mn}$. Under this
identification, the trace-symplectic product becomes the standard symplectic form,
\[
    \langle(x,z),(x',z')\rangle_{\sym}
    ~:=~
    \langle x,z'\rangle-\langle z,x'\rangle
    ~=~
    [(a,b),(a',b')] \;.
\]
We indicate the alphabet over which a quantum code is viewed by a subscript on its symplectic weight. Thus, the $q$-ary weight on $\F_p^{2mn}$ is
\[
    \wt_q(x,z)
    ~:=~
    \left|\left\{
        i\in[n]:(x_i,z_i)~\neq~(0^m,0^m)
    \right\}\right|,
\]
where $x_i,z_i\in\F_p^m$ are the blocks corresponding to the $i$th
$q$-ary coordinate.

\begin{definition}[Isotropic Subspaces]
An $\F_p$-linear subspace $L\leq\F_p^{2mn}$ is \emph{isotropic} if
$L\subseteq L^{\perp_{\sym}}$, where
\[
    L^{\perp_{\sym}}
    ~:=~
    \left\{
        v\in\F_p^{2mn}:
        \langle v,u\rangle_{\sym}~=~0
        \text{ for every }u\in L
    \right\}.
\]
\end{definition}
Every isotropic $L$ admits a choice of phases defining a commuting stabilizer
group, which defines the stabilizer code
\[ \mathcal{C}(L) = \{ |x\rangle \in (\mathbb{C}^q)^{\otimes n} \,:\, X(a)Z(b)|x\rangle = |x\rangle \, \forall (a,b) \in L \}.\] 
Conversely, every stabilizer code determines an isotropic subspace.
Thus, up to this choice of phases, isotropic subspaces and stabilizer codes are
the same objects; we will use the language of isotropic subspaces. Formally, we have the following definition.

\begin{definition}[Parameters of the Stabilizer Code]\label{def:stabcode}
Let $L \leq \mathbb{F}_p^{2mn}$ be an $\F_p$-linear isotropic subspace.
If $\dim_{\F_p}L=m(n-k)$, then, for $q=p^m$, $L$ defines an
$[[n,k,d]]_q$ stabilizer code that encodes $k$ qubits and has
distance:
\[
    d
    ~:=~
    \min\left\{
        \wt_q(v):
        v\in L^{\perp_{\sym}}\setminus L
    \right\}.
\]
The code has checks of weight at most $w$ if $L$ admits a basis of vectors, each with weight at most $w$. We say that (the code defined by) $L$ is QLDPC if $w = O(1)$, \ie a constant independent of $n$.
\end{definition}

\begin{definition}[Symplectic Maps]\label{def:symplectic}
A linear map $\phi:\F_p^n \to \F_p^m$ is \textit{symplectic} if for every pair of vectors $u,v \in\F_p^n$, we have $\ip{\phi(u)}{\phi(v)}_\sym = \ip{u}{v}_\sym$.  Moreover, let $\{u_i\}_i$ and $\{v_i\}_i$ be a pair of (symplectic) orthogonal bases for a pair of isotropic subspaces of the same dimension. Then, there exists a symplectic map that maps $u_i \to v_i$ for each $i$.
 \end{definition}

\paragraph{Notation.}
We collect the conventions used across the paper here.
\begin{itemize}[leftmargin=4mm,nosep]
\item We will use the notation $[[n,Rn,\delta n]]_p$ to denote a $p$-ary code  . The rate will be denoted by $r,R$, and the normalized distance by $\delta$, whereas the unnormalized versions of these will be denoted by $k$ and $d$ respectively. The subscripts $\out$, $\inn$, and $\concat$ will be used to distinguish the parameters of the outer,
    inner, and concatenated codes.
    \item For $q=p^m$, $\wt_q(v)$ counts the nonzero blocks in $(\F_p^{2m})^n$, while $\wt_p(v)$ counts the nonzero pairs in $(\F_p^2)^{mn}$. We write $\supp(v)$ for the corresponding set of nonzero coordinates; the alphabet will be clear from context.
    \item Given a subspace $L\leq \F_a^n$ over alphabet $a$ and a vector $v \in \F_a^n$, its \emph{coset weight} with respect to $L$ is
    \[
 \norm{v}_L
        ~:=~
        \min_{y\in L}\wt_a(v+y).
    \]
    \item  For any $w\in [n]$, and any subset $M \subseteq \F_q^n$, we define:
\[A_w(M) ~:=~ {\abs[\big]{\braces[\big]{y\in M \,\mid \,
    \wt_q(y) = w}}}, \; \text{ and analogously } A_{< w}(M).
    \] 
\end{itemize}

\subsection{Concatenation}

Concatenation encodes each coordinate of an \textit{outer code} over a large alphabet, say of size $q = p^m$, via an \textit{inner code} over alphabet $\F_p$. Concatenation was defined for classical codes by Forney~\cite{F66}, and generalized to quantum codes by~\cite{KL96,AB97,Zal97}. Following Hamada~\cite{Ham08}, we will use the following linear-algebraic formalization of concatenation for stabilizer codes.

\begin{description}
	\item[Outer Code] Let $L_{\out}\leq\F_p^{2mn}$ define an $[[n,k,d_{\out}]]_q$ outer code, where $q=p^m$.  Write vectors of $L_\out$ in block form
$u=(u_1,\ldots,u_n)$, with $u_i\in\F_p^{2m}$ corresponding to the $i^\mathrm{th}$
$q$-ary coordinate. 
	\item[Inner Codes] For each $i\in[n]$, let $L_i\leq\F_p^{2n_{\inn}}$ define an
$[[n_{\inn},m]]_p$ inner code, and denote $L_{\inn} := \oplus_i L_i $ Each inner code encodes one $q$-ary qudit within its logical space, \ie
$L_i^{\perp_{\sym}}/L_i$.

	\item[Encoding Maps]   Choose a subspace $W_i$  such that $L_i^{\perp_{\sym}}=L_i\oplus W_i$, and pick symplectic maps:
\[
    \phi_i:\F_p^{2m}\xrightarrow{\;\sim\;}W_i
    , \quad \; \Phi(u_1,\ldots,u_n)~:=~(\phi_1(u_1),\ldots,\phi_n(u_n))\,.
\]
\end{description}

\begin{definition}[Concatenated Stabilizer Code]\label{def:concatenated_code}
The concatenation of $L_{\out}$ with the inner encodings above is
\[
    L_{\concat}
    ~:=~
    L_{\inn}\oplus\Phi(L_{\out})
    ~\leq~ \F_p^{2nn_{\inn}}.
\]
\end{definition}

The intuition for \Cref{def:concatenated_code} is as follows.  The term $L_{\inn}$ represents the stabilizers defining $\mathcal{C}_1 \otimes \mathcal{C}_2 \otimes \cdots \otimes \mathcal{C}_n$, where $\mathcal{C}_i \subseteq (\mathbb{C}^p)^{\otimes n_{\inn}}$ is the $i^{\textrm{th}}$ inner code.  The term $\Phi(L_{\out})$ represents stabilizers defining the outer code: indeed, $L_{\out}$ defines the outer code; and $\Phi$ transforms these stabilizers so that they act on the inner codespaces.

\begin{proposition}[Symplectic dual of Concatenated subspace]\label{clm:concat_perp}
The concatenated subspace $L_\con$ is isotropic and satisfies:
\[
    L_{\concat}^{\perp_{\sym}}
    ~=~
    L_{\inn}
    \oplus
    \Phi(L_{\out}^{\perp_{\sym}}).
\]
Consequently, it defines an $[[nn_{\inn},mk,d]]_p$ stabilizer code, where
\[
    d
    ~=~
    \min_{y\in L_{\out}^{\perp_{\sym}}\setminus L_{\out}}
    \norm{\Phi(y)}_{L_{\inn}}  ~=~ \min_{y\in L_{\out}^{\perp_{\sym}}\setminus L_{\out}}\;\;
\sum_{i=1}^n \norm{\phi(y_i)}_{L_{i}}   \;.
\]
\end{proposition}

\begin{proof}
A vector $v\in\F_p^{2nn_{\inn}}$ is orthogonal to
$L_{\inn}$ if and only if it has a unique decomposition
\[
    v~=~\ell+\Phi(y),
    \qquad
    \ell\in L_{\inn},
    \quad
    y\in\F_p^{2mn},
\]
because $L_i^{\perp_{\sym}}=L_i\oplus W_i$ in every block. For
$u=(u_1,\ldots,u_n)\in L_{\out}$, the symplecticity of the maps $\phi_i$
gives
\[
    \langle v,\Phi(u)\rangle_{\sym}
    ~=~
    \sum_{i~=~1}^n
    \langle\phi_i(y_i),\phi_i(u_i)\rangle_{\sym}
    ~=~
    \langle y,u\rangle_{\sym}.
\]
Hence $v$ is also orthogonal to $\Phi(L_{\out})$ if and only if
$y\in L_{\out}^{\perp_{\sym}}$, proving the dual-space identity. Since
$L_{\out}\subseteq L_{\out}^{\perp_{\sym}}$, this identity also proves that
$L_{\concat}$ is isotropic. The sum defining $L_{\concat}$ is direct, and therefore
\[
    \dim_{\F_p}L_{\concat}
    ~=~
    n(n_{\inn}-m)+m(n-k)
    ~=~
    nn_{\inn}-mk.
\]
Thus $L_{\concat}$ defines an $[[nn_{\inn},mk]]_p$ code. Finally, the unique
decomposition above shows that a vector of
$L_{\concat}^{\perp_{\sym}}$ lies outside $L_{\concat}$ precisely when
$y\notin L_{\out}$, which gives the stated distance formula.\qedhere
\end{proof}

 Both our random models will utilize concatenation and can be seen as a quantum version of the classical construction of Thommesen~\cite{T83}. We use a ``nice'' outer stabilizer code and output a ``random concatenation`` of this outer code. Stated more precisely, we pick a ``nice'' isotropic subspace $L_\out \leq \F_p^{2mn}$, and output $\Phi(L_\out)\oplus L_\inn$, where $\Phi$ and $ L_\inn$ are picked randomly.

\section{Random Model 1: Trivial Inner Codes}\label{sec:model1}
We start with the simplest form of concatenation in which $L_\inn=\{0\}$, and thus, we only have the freedom to pick the encoding functions, $\phi_i:\F_p^{2m}\to \F_p^{2m}$ randomly. The random model picks each of these functions, independently and uniformly from the set of symplectic isomorphisms. The final concatenated code is then:
\begin{align*}
     L_{\concat}~=~\Phi(L_{\out})~=~ \braces[\big]{(\phi_1(v_1),\cdots, \phi_n(v_n)) \mid (v_1,\cdots, v_n) \in L_\out}\,.
\end{align*}

The goal of this section is to prove the following key property of this random ensemble and derive its consequences:
\begin{restatable}{theorem}{rateone}\label{thm:final_rate1} Fix a prime $p$, and a constants $R,\gamma > 0$. Then for every large enough $N$, there exists an explicit stabilizer subspace $L_\out$, such that the ensemble of stabilizer codes, $\{\Phi(L_\out)\}_\Phi$,  has blocklength $N$, rate $R$, and checks of weight at most $O_\gamma(1)$. Moreover, for any fixed vector $z \in \F_p^N$, we have: 
\begin{equation}    \Pr{\Phi}{z\in
    \Phi(L_\out)^{\perp_{\sym}}\setminus \Phi(L_\out)}
    ~\leq~
    p^{-N{(1-R -\gamma)}}\;.
\end{equation}
As a corollary, a random code from this ensemble approaches the GV bound, $d \geq (\delta_{GV} - \eps(\gamma))\cdot N$, with probability at least $1-\exp(-\Omega_\gamma(N))$.
\end{restatable}

This special case of $L_\inn = 0$ is easier to analyze because from~\cref{clm:concat_perp}, we get that the non-trivial vectors of the concatenated code are just the encoding of the non-trivial vectors of the outer code, \ie  $L_{\concat}^\perp \setminus L_{\concat} =  \Phi(L_{\out}^{\perp_{\sym}}\setminus L_\out)$. A simple consequence of having $L_\inn =0$ is that the distance of our concatenated code is the smallest value of $\wt_p\parens[\big]{\Phi(y)}$ over all choices of $y\in L_\out^\perpsym\setminus L_\out$. Observe that for any $\Phi$, $\wt_p\parens[\big]{\Phi(y)}$ only depends on the support of $y$, which motivates the following useful notation. For any $S\subseteq [n]$, and any subset $M \subseteq \F_q^n$:
\[A_S(M) ~:=~ {\abs[\big]{\braces[\big]{y\in M \,\mid \,
    \{i:y_i\neq0\}~=~S}}}.
    \]
The proof of \cref{thm:final_rate1} consists of two steps: (i) (\cref{prop:tentative_fixed_vector,lem:eta_outer_enum}) showing that the probability that any fixed vector $z$ is a non-trivial codeword of the concatenated code can be bound using the weight distribution of the outer code $A_{S}(L_{\out})$. This can be controlled if the outer code satisfies a quantitative property that we call $\eta$-degeneracy, and (ii) constructing QLDPC codes with the required degeneracy property. In this section we prove the first step and relegate the second to~\cref{sec:AEL}.

\begin{proposition}\label{prop:tentative_fixed_vector}
Let $z=(z_1,\ldots,z_n)\in\F_p^{2mn}$ be fixed and define $S=\{i:z_i\neq0\}$. Then
\[
    \Pr{\Phi}{z\in
    L_{\concat}^{\perp_{\sym}}\setminus L_{\concat}}
    ~=~
  (p^{2m}-1)^{-|S|}  \cdot A_S(L_{\out}^{\perp_{\sym}}\setminus L_{\out})\,.
\]
\end{proposition}

\begin{proof}
Since, $L_{\concat}^\perp \setminus L_{\concat} =  \Phi(L_{\out}^{\perp_{\sym}}\setminus L_\out)$, we have that for any $z$,
\[
z\in
    L_{\concat}^{\perp_{\sym}}\setminus L_{\concat} \iff \Phi^{-1}(z) = (\phi_1^{-1}z_1, \dots, \phi_n^{-1}z_n) \in
    L_{\out}^{\perp_{\sym}}\setminus L_{\out}\;.
\]
For each $i\in S$, the symplectic group acts transitively on the nonzero vectors of $\F_p^{2m}$. Hence, $\phi_i^{-1}(z_i)$ is uniform over $\F_p^{2m}\setminus\{0\}$, independently across $i\in S$, while $\phi_i^{-1}(0)=0$. Thus, $\Phi^{-1}(z)$ is uniform over the $(p^{2m}-1)^{|S|}$ vectors with block support exactly $S$. 
\end{proof}

\begin{lemma}[Weight bounds via Outer code]\label{lem:eta_outer_enum}
Let $L\leq\F_p^{2mn}$ define an $[[n,k,d]]_q$ stabilizer code and fix an integer $1\leq d_0\leq d$. Then, for any set $S\subset[n]$:
\[
    A_S(L^{\perp_{\sym}}\setminus L)
    ~\leq~
   (p^{2m})^{|S|-d_0+1}\cdot A_{<d_0}(L)
    \;.
\]
\end{lemma}
\begin{proof}
If $|S| < d$, then the inequality trivially holds as $A_S=0$ by the definition of the distance of the code. Otherwise, fix
any $T\subseteq S$ of size $|S|-d_0+1$. Group the
logical vectors with support $S$ by their restriction to $T$, giving at most
$(p^{2m})^{|S|-d_0+1}$ groups. Two vectors in one group differ by an element of
$L^{\perp_{\sym}}$ supported on $S\setminus T$, whose weight is below $d_0$
and hence below $d$; therefore, their difference  lies in $L$. Translating a
group by one of its elements injects it into
$\{v\in L:\wt_q(v)<d_0\}$, so each group has size at most
$A_{<d_0}(L)$. 
\end{proof}

Motivated by \cref{lem:eta_outer_enum}, we introduce the notion of \emph{$\eta$-degeneracy} to get a handle on $A_{d}(L)$. Informally, such a code is one where the stabilizer group $L$ does not contain too many low-weight stabilizers.

\begin{definition}[Low-weight degeneracy]\label{def:eta_degenerate}
Let $L\leq\F_p^{2mn}$ define an $[[n,k,d]]_q$ stabilizer code. Fix an integer $1\leq d_0\leq d$. For $0\leq\eta\leq1$, the code is
\emph{$\eta$-degenerate below $d_0$} if
\[
    A_{<d_0}(L)
    ~:=~
    \left|\{\ell\in L:\wt_q(\ell)<d_0\}\right|
    ~\leq~
    |L|^\eta.
\]
We simply call the code $\eta$-degenerate when this holds for $d_0=d$.
\end{definition}

Being $0$-degenerate is equivalent to a code being \emph{pure}, whereas every code is trivially $1$-degenerate. The proof of Ouyang~\cite{Ouy14} crucially utilized the purity property of the outer code, namely Reed-Solomon codes. However, a QLDPC family with distance growing with $n$ can never be pure. In fact, for a QLDPC code, $L$ is generated by constant-weight vectors; it contains multiple low-weight vectors. The key idea is that despite this issue, QLDPC codes can satisfy this relaxed notion of $\eta$-degeneracy, which suffices to bound the number of low-weight vectors.

\begin{lemma}[Degeneracy implies weight bounds]\label{cor:degen_cor}
For any $R,\gamma >0$, there exists constants $\eta,\zeta \in (0,1)$ with the following property. Let $L_\out \leq\F_p^{2mn}$ be an $[[n,Rn,d]]_{p^m}$ stabilizer code that is $\eta$-degenerate below $d_0 = (1-\zeta) \frac{(1-R)n}{2}$, for large enough $n,m$. Then, the collection $\{L_{\concat} =\Phi(L_\out)\}_\Phi$ is an ensemble of $[[mn, Rmn, \delta mn]]_p$ codes such that for any nonzero vector $z$,
\begin{equation}\label{eq:tentative_uniform_fixed_vector}
    \Pr{\Phi}{z\in
    L_{\concat}^{\perp_{\sym}}\setminus L_{\concat}}
    ~\leq~
    p^{-mn{(1-R-\gamma)}}\;.
\end{equation}
For every $\eps>0$, choosing $\gamma$ sufficiently small gives $\delta\geq\deltagv(R)-\eps$ with probability at least $1-\exp(-\Omega_\eps(n))$.
\end{lemma}
\begin{proof} Pick $\eta, \zeta < \frac{\gamma}{3(1-R)}$. From~\cref{lem:eta_outer_enum} and \cref{def:eta_degenerate}, we get,
\[  A_S(L_\out^{\perp_{\sym}}\setminus L_\out)
     ~\leq~
    (p^{2m})^{w-d_0+1}|L_\out|^\eta.
\] Using this along with~\cref{prop:tentative_fixed_vector}, and using $|S|\leq n$ gives:
\[
\Pr{\Phi}{z\in
    L_{\concat}^{\perp_{\sym}}\setminus L_{\concat}}
    ~\leq~ |L_{\out}|^\eta p^{-2m(d_0-1)}(1-p^{-2m})^{-n}\,.
\] 
Plugging in the lower bound for $d_0,\eta, \zeta$ gives us our first claim for large enough $m,n$:
\begin{align*}\label{eq:tentative_fixed_vector_exponent}
    \frac{1}{mn}\log_p
    \Pr{\Phi}{z\in
    L_{\concat}^{\perp_{\sym}}\setminus L_{\concat}}
    ~&\leq~
    -(1-\zeta-\eta)(1-R)
    +\frac{2}{n}
    +\frac{1}{m}\log_p\frac{1}{1-p^{-2m}}. \\
    ~&\leq~
    -\parens[\big]{1-R - \frac{2\gamma}{3} + o_{m,n}(1)}\; .
\end{align*}
Finally, for the GV bound, we pick $2\gamma <1-R-2H_{p^2}(\delta)$ where $\delta
    ~:=~ \deltagv(R)-\eps$.
The number of nonzero vectors of $p$-ary weight at most $\delta mn$ is at most $(p^2)^{mnH_{p^2}(\delta)}$. Using~\cref{eq:tentative_uniform_fixed_vector} and taking a union bound over all these low-weight vectors $z$ gives:
\[
    \Pr{\Phi}{d(L_{\concat})\leq\delta mn}
    ~\leq~
    p^{-mn(1-R-\gamma-2H_{p^2}(\delta))}
    ~\leq~     p^{-\gamma mn} \;. \qedhere
\]
\end{proof}

\cref{cor:degen_cor} thus proves~\cref{thm:final_rate1} assuming the existence of codes that are $\eta$-degenerate for an arbitrarily small (but constant) $\eta$. Additionally, they also need to have their distance be arbitrarily close to the quantum Singleton bound. These are both very strong properties, and we will show the construction of such codes in~\cref{sec:AEL}. For now, let us compare our ensemble to the ensemble of all stabilizer codes:
\begin{align*}
(\text{Our } O_\gamma(1)-\text{QLDPC Ensemble})\qquad 	 \Pr{\Phi}{z\in
    L_{\concat}^{\perp_{\sym}}\setminus L_{\concat}}
    ~&\leq~
    p^{-N{(1-R-\gamma)}}\\
(\text{All Stabilizer Codes})\qquad    	 \Pr{\Phi}{z\in
    L_{\stab}^{\perp_{\sym}}\setminus L_{\stab}}
    ~&\leq~
    p^{-N{(1-R)}}
\end{align*}
Thus, what~\cref{thm:final_rate1} gives us is a random ensemble that is $\gamma$-close to a truly random stabilizer code in the sense of this pointwise bound. Because of this pointwise bound, we will now show that codes from our random ensemble achieve a rate almost equal (up to a $\gamma$-loss) to a random stabilizer code on various standard channels such as the \textit{erasure channel} and the \textit{depolarizing channel}. The upshot is that the random ensemble is QLDPC with check weights only depending on $\gamma$. Throughout the next section, we will refer to this ensemble as our \emph{rate-one} ensemble for brevity.

\subsection{Performance on Quantum Channels}

\paragraph{Erasure Channel.} The easiest channel to analyze is the erasure channel in which the channel erases each qudit independently with probability $\kappa$, and with probability $(1-\kappa)$ leaves it unchanged. Bennet, DiVincenzo, and Smolin~\cite{BDS97} proved that the capacity of the quantum erasure channel is $1-2\kappa$, which is achieved by a random stabilizer code. We now show that for any $\eps >0$, one can achieve a rate of $1-2\kappa-\eps$ with our random ensemble of $O_\eps(1)$-QLDPC codes.

\begin{corollary}[Erasure-channel capacity]\label{cor:tentative_erasure}
Let $\kappa <1/2$ and $\eps >0$ be any fixed constants. Then, for every large enough $N$, there exists an $O_{\eps}(1)$-QLDPC code of rate $R > 1-2\kappa-\eps$ that can decode the erasures caused by the $p$-ary quantum erasure channel with erasure probability $\kappa$, with decoding error $\exp(-\Omega_\eps(n))$. 
\end{corollary}

\begin{proof} Let $\{L_\con\}$ be the family of codes obtained from~\cref{thm:final_rate1} for $R = 1-2\kappa-\eps/2$, and $\gamma <\eps /8$. Let $\tau = \kappa + \gamma$. By conditioning, we get that the probability of a decoding error is 
\begin{align*}
	\Pr{\Phi}{ \text{ Error }}
    ~&=~ \sum_{E\subset [N]}\Pr{\Phi}{L_\con \text{ fails to decode } \mid E \text{ is the erased set}} \;\cdot\; \Pr{\kappa}{E \text{ is the erased set}}\\
     ~&\leq~ \max_{|E|\leq \tau N}\Pr{\Phi}{L_\con \text{ fails to decode } \mid E \text{ is the erased set}}  + \sum_{|E| > \tau N} \Pr{\kappa}{E \text{ is the erased set}}\;\\
     ~&\leq~ \max_{|E|\leq \tau N}\Pr{\Phi}{L_\con \text{ fails to decode } \mid E \text{ is the erased set}}  + \exp^{-\Omega_\eps(N)}.
\end{align*}
The last inequality here follows from the Chernoff bound. Now, let $E\subseteq[N]$ be a fixed erased set. Once $E$ is known, the decoder only has to correct errors supported on $E$, which is possible iff:
\[
    \bigl(L_{\concat}^{\perp_{\sym}}\setminus L_{\concat}\bigr)
    \cap \F_p^{2E}
    ~=~
    \emptyset,
\]
where $\F_p^{2E}$ denotes the vectors in $\F_p^{2n}$ supported on $E$. For fixed $E$ of size $s$, there are at most $p^{2s}$ vectors supported on $E$. Hence, using~\cref{eq:tentative_uniform_fixed_vector} and a union bound over these vectors,
\[
   \max_{|E|\leq \tau N}\Pr{\Phi}{L_\con \text{ fails to decode } \mid E \text{ is the erased set}}
    ~\leq~
    p^{2\tau N-N(1-R-\gamma)}~\leq~p^{-\frac{\eps}{8}N}. \qedhere
\]
\end{proof}

\paragraph{Memoryless Pauli Channels.} We have the $N$-qudit $p$-ary Pauli group indexed by $\{P(a)\mid a \in \F_p^{2N}\}$. A \textit{memoryless Pauli channel} is defined by a distribution, $\nu:\F_p^{2} \to [0,1]$, such that each $P(a)$ acts on the state with probability $\prod _ {i=1}^{N}\nu(a_i)$. 
In other words, the channel independently applies a Pauli $P(a_i)$ to $i^{\mathrm{th}}$ qudit with probability $\nu(a_i)$. An example of such a channel is the well-studied \textit{depolarizing channel},  which maps a $1$-qudit \textit{density matrix}\footnote{A density matrix represents a mixture of quantum states. So, it is  of the form $\sum_{a \in \F_p^{2a}} \alpha_a \ketbra{\psi_a}{\psi_a} $, where $\psi_a$ is the qudit labeled by $a$.} $\rho$ as follows:
\[
\rho  \to (1-\eps)\rho  ~+~ \frac{\eps}{p^2-1} \sum_{b\in\F_p^2\setminus\{0\}} P(b)\cdot \rho \cdot P(b)^\dagger
\]

Therefore, this channel is a Pauli channel with $\nu(0) =  1-\eps$, and $\nu(b) ~=~ \frac{\eps}{p^2-1}$ for each non-zero $b \in \F_p^2$.  The quantum capacity is not known even for the depolarizing channel, which is one of the simplest memoryless Pauli channels. However, an achievable lower bound, called the \textit{Hashing bound} or the \textit{random coding bound} is known, which says that there exists a stabilizer code that can transmit information at a rate of at least $1-H_2(\nu)/\log_2 p$, where:
\[
    \mathsf H_2(\nu)
    ~:=~
    -\sum_{b\in\F_p^2}\nu(b)\log_2 (\nu(b))\;.
\] 
We will now show that an $O_\eps(1)$-QLDPC code obtained from our ensemble can transmit information across the Pauli channel defined by $\nu$ at rate at least $1-H_p(\nu)-\eps$. 

\begin{definition}[$\eta$-typical Pauli errors]
Let $\nu$ be a probability distribution on the single-qudit Pauli alphabet \(\F_p^2\). For \(z=(z_1,\ldots,z_N)\in(\F_p^2)^N\), let the empirical distribution of $z$ be:
\[ \nu_z(b)  ~:=~
    \frac{1}{N}\abs{\{i\in{[N]}:z_i=b\}},
    \quad b\in\F_p^2 \;.
\]
The \(\xi\)-typical set of
Pauli errors with respect to \(\nu\) is
\begin{align*}
T_\xi(\nu)
    ~&:=~
    \braces[\big]{z\in \F_p^{2N}:
        D_{\mathrm{KL}}(\nu_z \Vert \nu)\leq \xi}, \;
         \text{ where },\;  D_{\mathrm{KL}}(\mu \Vert \nu) ~=~ \sum_{b\in \F_p^2} \mu(b) \log_2 \frac{\mu(b)}{\nu(b)}	\;.
\end{align*}
\end{definition}
We state the following useful fact without proof and refer the reader to~\cite{CoverT06}. The statement justifies the name \enquote{typical set} for $T_\xi$.

\begin{fact}[Concentration for typical vectors, Sanov's theorem]\label{fact:sanov}
Let $T_\xi(\nu)$ be the set of typical vectors as defined above. Then, there exists $\alpha(\xi)\to0$ as $\xi\to0$ such that for large enough $N$,
\[
    |T_\xi(\nu)|
    ~\leq~
    2^{N(\mathsf H_2(\nu)+\alpha(\xi))},
    \quad \text{ and } \quad 
    \Pr{x\sim \nu^{\otimes N}}{x\notin T_\xi(\nu)}
    ~\leq~
    2^{-N\cdot (\xi + o_N(1))}.
\] 
\end{fact}

\begin{corollary}[Hashing bound for randomly concatenated codes]\label{cor:tentative_hashing}
Consider the memoryless Pauli channel with single-qudit error distribution $\nu$. For every $\eps >0$, and  $ R\leq 1- \frac{H_2(\nu)}{\log_2(p)} -\eps$, a random code in the rate-one ensemble has decoding error $\exp(-\Omega(N))$ with probability $1-\exp(-\Omega(N))$. 
\end{corollary}

\begin{proof}
Choose $\xi,\gamma>0$ such that $\alpha(\xi)/\log_2(p)+\gamma< \eps$, and let $T_\xi(\nu)$ be the $\xi$-typical set of Pauli errors as defined above. For each syndrome, the decoder chooses a fixed representative in $T_\xi(\nu)$, if one exists. We count all atypical errors as failures. Now, fix an error pattern $e\in T_\xi(\nu)$ that occurs. If the decoder fails on $e$, then its chosen representative $e'\in T_\xi(\nu)$ satisfies $e-e'\in L_\con^\perpsym\setminus L_\con$. Now, we can use~\cref{eq:tentative_uniform_fixed_vector} to union bound over all such bad events, and we get:
\[
\max_{e\in T_\xi(\nu)}\Pr{\Phi}{L_\con \text{ fails to decode } \mid e \text{ is the error}} ~\leq~ |T_\xi(\nu)|\cdot p^{-N(1-R-\gamma)} ~\leq~ p^{-\Omega_\eps(N)}\;. 
\] 
The final inequality above used the bound from~\cref{fact:sanov}. 
Let $P_{\mathrm{fail}}(\Phi)$ denote the decoding error for the code defined by $\Phi$. Averaging over errors gives:
\begin{align*}
	\Ex{\Phi}{P_{\mathrm{fail}}(\Phi)}
    ~&=~ \sum_{e}\nu^{\otimes N}(e)\cdot\Pr{\Phi}{L_\con \text{ fails to decode } \mid e \text{ is the error}}\\
     ~&\leq~ \max_{e\in T_\xi(\nu)}\Pr{\Phi}{L_\con \text{ fails to decode } \mid e \text{ is the error}} + \Pr{e\sim\nu^{\otimes N}}{e\notin T_\xi(\nu)}\\
     ~&\leq~ p^{-\Omega_\eps(N)} + p^{-\Omega_\eps(N)} ~\leq~ p^{-\Omega_\eps(N)}.
\end{align*}
By Markov's inequality, $P_{\mathrm{fail}}(\Phi)\leq p^{-\Omega_\eps(N)}$ with probability at least $1-p^{-\Omega_\eps(N)}$, as claimed.
\end{proof}

\paragraph{General memoryless channels.}
We now work with arbitrary memoryless channels that map between density matrices. These channels act by operators of the form $\cN^{\otimes N}$, where each $\cN$ acts on a single-qudit density matrix by $\mathcal N (\sigma) = \sum_j K_{j} \sigma K_{j}^\dagger$. These operators, $K_j$, called the Kraus operators, satisfy $\sum_j K_j^\dagger K_j=I$ and can be expanded in the single-qudit Pauli basis:
\[
    K_j~=~\sum_{a\in\F_p^2}c_{j,a} P(a),
    \quad \implies \quad \mathcal N (\sigma) ~= \sum_j\sum_{a,b\in\F_p^2} c_{j,a}\bar{c}_{j,b}\, P(a)\sigma P(b)^\dagger\;.
\]

Notice that the off-diagonal terms ($a\neq b$) in the above expression make the channel non-Pauli. The key idea (due to~\cite{Ham03}) is to reduce the analysis to the Pauli part of the channel, \ie the diagonal terms:
\[   \mathcal N_P(\sigma) ~:=~ \sum_{a} c_a \,P(a)\sigma P(a)^\dagger, \quad\text{where~~}
c_a~:=~\sum_j|c_{j,a}|^2.
\]
Since $\Tr(P(a)^\dagger P(b))= p\cdot\mathbf 1_{a=b}$, we have $\sum_a c_a=p^{-1}\sum_j\Tr(K_j^\dagger K_j)=1$. Thus $P_{\mathcal N}(a):=c_a$ is a probability distribution and $\mathcal N_P$ is a Pauli channel.

We will use \textit{entanglement fidelity}\footnote{Entanglement fidelity is at most the average pure-state fidelity over any orthonormal basis of $\cC$, but need not be a lower bound for the worst-case pure-state fidelity.}~\cite{Sch96} as our notion of how well the code can correct errors. We use the following equivalent definition, which is easier to work with; see~\cite[Thm.9.5.1]{Wil17}.

\begin{definition}[Entanglement Fidelity]
	For a density matrix $\sigma$ and a channel $\cQ$ given by Kraus operators $\{K_j\}_j$, 
	\[
	F_e(\sigma,\cQ) ~:=~ \sum_j|\Tr(\sigma K_j)|^2 . \]
For a code subspace $\cC$ with orthonormal basis $\mathcal B$, we use the maximally mixed state on $\mathcal C$:
	\[
	\sigma(\cC)~:=~\frac1{\dim\cC}\sum_{v\in\mathcal B}\ketbra{v}{v},
	\qquad F_e(\cC,\cQ)~:=~F_e(\sigma(\cC),\cQ).
	\]
\end{definition}
Entanglement fidelity measures how correlated the output of the channel is to the input state. For example, if $\cQ$ is the identity map, \ie leaves every state unchanged, then $F_e(\cC, \cQ) = 1$. Intuitively, a code $\cC$ is a good error-correcting code over a noise channel $\cN$ if there exists a recovery map $\cR$ such that $F_e(\cC, \cR\circ \cN) \approx 1$. A recovery map is closely related to the notion of a decoder, as one may define it to be a channel that maps a pure state to its decoded state. For such a recovery map, one has the simple inequality,
\[
F_e(\cC,\cR\circ\cN) ~\geq~ 1-\Pr{\,\text{Decoder fails for~}\cN\,}\;.
\]

The key upshot of working with this notion is that it is easy to average and enables a reduction to analyzing the Pauli part of the channel. The setup is as follows. Let $\cN^{\otimes N}$ be a given memoryless channel, and let $\cN_P^{\otimes N}$ be its Pauli part. We will define a family of subspaces $\{C_s\}$, and corresponding recovery maps $\{\cR_s\}_s$ such that (via a result due to Hamada~\cite{Ham03}): 
\[
  \Ex{s}{F_e\bigl(\mathcal C_{s},
    \mathcal R_{s}\circ\mathcal N^{\otimes N}\bigr)}
    ~=~
   F_e\bigl(\mathcal C_{0},
    \cR_{0}\circ\mathcal N_P^{\otimes N}\bigr) .
 \]
Therefore, if one has a good error-correcting code $\cC_0$ and a recovery map $\cR_0$ over the Pauli channel $N_P^{\otimes N}$, one also has it for the entire channel $\cN$ using some $\cC_s, \cR_s$. Now, to define these spaces, we will work with the shifts of a stabilizer code. Fix a basis $\{u_1,\cdots u_t\}$ of a stabilizer space $L\leq\F_p^{2N}$. Let $s=(s_1,\ldots,s_t)\in\F_p^t$ and define:
\[
    \mathcal C_{s}
    ~:=~
    \left\{|\psi\rangle:
    P(u_j)|\psi\rangle~=~\omega_p^{s_j}|\psi\rangle
    \text{ for every }j\in[t]\right\}.
\]
Note that $\cC_0$ is the stabilizer code corresponding to the space $L$.
Let $\cR_0$ be some recovery channel (that corrects over the Pauli noise channel $\cN_P^{\otimes N}$), and let its Kraus operators be $\{B_i(\cR_0)\}_i$. For each $s \in \F_p^{t}$, choose a vector $v_s$ such that for each $i$, $\ip{v_s}{u_i}_\sym = s_i$, and let $\Pi_s=P(v_s)$. Then, the density matrix and the recovery channel for the shifted code $C_s$ can be explicitly defined as: 
\begin{equation}\label{eq:shift_s}	\cC_s=\Pi_s\cC_0,\qquad \sigma(\cC_s)=\Pi_s\,\sigma(\cC_0)\,\Pi_s^\dagger, \qquad B_i(\cR_s)~:=~\Pi_s\,B_i(\cR_0) \,\Pi_s^\dagger.
\end{equation}

We will now state and prove Hamada's lemma that reduces the (average) entanglement fidelity of the ensemble $\cC_s$ over a general channel to that of $\cC_0$ over a Pauli channel. We present a slightly rephrased version of Hamada's argument below to keep the section self-contained. 

\begin{lemma}[\cite{Ham03}] Let $s$ be uniform in $\F_p^t$, and $C_s, \cN, \cN_P$ be as defined above. Let $\cR_s$ be the recovery channel with Kraus operators $\{B_i(\cR_s) = \Pi_s B_i\Pi_s^\dagger\}_i$, where $\{B_i\}_i$ are the Kraus operators of $\cR_0$. Then,
\begin{equation}\label{eq:tentative_hamada_identity}
    \Ex{s}{F_e\bigl(\mathcal C_{s},
    \mathcal R_{s}\circ\mathcal N^{\otimes N}\bigr)}
    ~=~
   F_e\bigl(\mathcal C_{0},
    \cR_{0}\circ\mathcal N_P^{\otimes N}\bigr) .
\end{equation}
\end{lemma}
\begin{proof}
Write $\sigma=\sigma(\cC_0)$, let $\{B_i\}_i$ be the Kraus operators for the recovery channel $\cR_0$, and let $\{K_j\}$ be the Kraus operators of $\cN^{\otimes N}$. Plugging in our definition of entanglement fidelity, we get,
\begin{align*}
\Ex{s}{F_e(\cC_s,\cR_s\circ\cN^{\otimes N})}&\quad=~\Ex{s}{\sum_{i,j}\left|\Tr\left(\sigma(\cC_s)\cdot B_i(\cR_s)\cdot  K_j\right)\right|^2} \\
&\quad=~\sum_{i,j}\Ex{s}{\left|\Tr\left(\Pi_s\sigma\Pi_s^\dagger\Pi_sB_i\Pi_s^\dagger K_j\right)\right|^2}&&({\text{\cref{eq:shift_s}}})\\
&\quad=~\sum_{i,j}\Tr\left[(\sigma B_i\otimes\overline{\sigma B_i})\,
\Ex{s}{(\Pi_s^\dagger K_j\Pi_s)\otimes\overline{\Pi_s^\dagger K_j\Pi_s}}\right]
\end{align*}
The last equality uses cyclicity of trace and $|\Tr(AB)|^2=\Tr\parens[\big]{(A\otimes\bar A)(B\otimes\bar B)}$. The fidelity does not depend on the choice of $v_s$, and we may therefore choose $v_s$ uniformly among vectors with syndrome $s$. Together with uniform $s$, this makes $v_s$ uniform in $\F_p^{2N}$. Write $K_j=\sum_a c_{j,a}P(a)$. Since $P(v)^\dagger P(a)P(v)=\omega_p^{\ip{v}{a}_\sym}P(a)$, the average inside the trace becomes
\begin{align*}
\Ex{v}{(P(v)^\dagger K_jP(v))\otimes\overline{P(v)^\dagger K_jP(v)}}
&=\sum_{a,b}c_{j,a}\overline{c_{j,b}}\,P(a)\otimes\overline{P(b)}\cdot\Ex{v}{\omega_p^{\ip{v}{a-b}_\sym}}\\
&=\sum_a |c_{j,a}|^2P(a)\otimes\overline{P(a)}.
\end{align*}
Substituting above and summing over $j$ gives $F_e(\cC_0,\cR_0\circ\cN_P^{\otimes N})$.
\end{proof}

 Using the above lemma, we can now analyze the entanglement fidelity for general memoryless channels, using the bounds for memoryless Pauli channels developed in~\cref{cor:tentative_hashing}.

\begin{corollary}[General memoryless channels]\label{cor:tentative_general_channel}
Consider the rate-one ensemble, $\{L_\con(\Phi)\}_\Phi$  of~\cref{cor:tentative_hashing}. Fix a basis $\{u_1,\ldots,u_t\}$ of $L_{\concat}(\Phi)$.
Choose $s=(s_1,\ldots,s_t)\in\F_p^t$ uniformly at random, and define,
\[
    \mathcal C_{\Phi,s}
    ~:=~
    \left\{|\psi\rangle:
    P(u_j)|\psi\rangle~=~\omega_p^{s_j}|\psi\rangle
    \text{ for every }j\in[t]\right\}.
\]
For every $s$, $\mathcal C_{\Phi,s}$ is a rate $R$ code as $C_{\Phi,0}$ is. If $ R~<~1-\mathsf H_p(P_{\mathcal N})$,
then there are recovery maps $\mathcal R_{\Phi,s}$ such that with probability $1-\exp(-\Omega(n))$, a random pair $(\Phi,s)$ satisfies
\[
    F_e\bigl(\mathcal C_{\Phi,s},
    \mathcal R_{\Phi,s}\circ\mathcal N^{\otimes mn}\bigr)
    ~\geq~
    1-\exp(-\Omega(n))\mper
\]
\end{corollary}

\begin{proof}
Apply~\cref{cor:tentative_hashing} to the Pauli channel with distribution $P_{\mathcal N}$. Its expected decoding error over $\Phi$ is $\exp(-\Omega(N))$. By~\cref{eq:tentative_hamada_identity},
\[
\Ex{\Phi,s}{1-F_e(\cC_{\Phi,s},\cR_{\Phi,s}\circ\cN^{\otimes N})}
~\leq~\Ex{\Phi}{\Pr{\text{Decoder fails for~}\cN_P^{\otimes N}\mid\Phi}}
~=~\exp(-\Omega(N)).
\]
Now one can apply Markov's inequality to prove the claim. In particular, for every $\Phi$ with small Pauli decoding error, there is an $s$ whose entanglement fidelity for $\cN^{\otimes N}$ is close to $1$.
\end{proof}

\section{Random Model 2: General Inner Codes}
Now we tackle the more general case where for
each $i\in[n]$, $L_i\leq \F_p^{2n_{\inn}}$ is a uniformly random isotropic subspace of dimension $n_{\inn}-m$. The encoding maps  are also picked uniformly randomly among the symplectic isomorphisms between
\[ L_i^\perpsym = L_i\oplus W_i,\quad	   \phi_i:\F_p^{2m}\to W_i.
\]
The concatenated code is thus $\Phi(L_\out) \oplus L_\inn$, and from~\cref{clm:concat_perp}, the space of non-trivial vectors is $\Phi(L_\out^{\perpsym}\setminus L_\out) \oplus L_\inn$. The first important difference from the simpler model earlier is that now we need to consider the \textit{coset weight} of an encoded vector, $\Phi(y)$, defined as:
\[    \norm{\Phi(y)}_{L_\inn}
    ~=~
    \sum_{i~=~1}^n \norm{\phi_i(y_i)}_{L_i},\quad \text{ where } \;\norm{v_i}_{L_i} ~:=~ \min_{a \in L_i} \norm{v_i +a}.
\]

To analyze this coset weight, we first compute the probability that a random $\phi_i$ maps any specific $y_i$ to any specific coset of a random inner code $L_i$. An analogous claim was shown by Ouyang~\cite[Proposition III.3]{Ouy14}.

\begin{lemma}[Uniformity of a single coordinate]
\label{lem:one_inner_coset} Let $L\subseteq \F_p^{2n_\inn}$ be a uniformly random isotropic subspace of dimension $n_{\inn}-m$, and let $\phi$ be a uniformly random symplectic isomorphism from $\F_p^{2m} \to L^\perpsym/L $.  For every fixed nonzero $y\in\F_p^{2m}$ and every fixed nonzero $z\in\F_p^{2n_{\inn}}$, 
\[
    \Pr{\phi, L}{z\in L+\phi(y)}
    ~=~\frac{p^{n_{\inn}-m}}{p^{2n_{\inn}}-1}
    ~=~\frac{p^{-(n_{\inn}+m)}}{1-p^{-2n_{\inn}}} .
\]
\end{lemma}

\begin{proof}
Since
$z\neq 0$ and $L_{}+\phi(y)\subseteq
L^{\perp_{\sym}}\setminus L$,
\abovedisplayskip=3pt
\[
    \Pr{}{z\in L+\phi(y)}
    ~=~
    \Pr{}{z\in L+\phi(y)\mid
        z\in L^{\perp_{\sym}}\setminus L}\cdot
    \Pr{}{z\in L^{\perp_{\sym}}\setminus L}.
\]
We compute the two factors separately. Conditional on $L$, the image of
the fixed $y\neq 0$ under a uniformly random symplectic isomorphism is uniform
over the $p^{2m}-1$ nonzero elements of
$L^{\perp_{\sym}}/L$. Whenever
$z\in L^{\perp_{\sym}}\setminus L$, exactly one of these
elements is the coset $z+L$. Therefore,
\[
    \Pr{}{z\in L+\phi(y)\mid
        z\in L^{\perp_{\sym}}\setminus L}
    ~=~\frac{1}{p^{2m}-1}.
\]
For the second factor, the uniform distribution on isotropic subspaces is
invariant under the symplectic group, which acts transitively on nonzero
vectors. Thus, double counting gives
\[
    \Pr{}{z\in L^{\perp_{\sym}}}
    ~=~
    \frac{p^{n_{\inn}+m}-1}{p^{2n_{\inn}}-1},
    \qquad
    \Pr{}{z\in L}
    ~=~
    \frac{p^{n_{\inn}-m}-1}{p^{2n_{\inn}}-1},
\]
because the corresponding dimensions are $n_{\inn}+m$ and $n_{\inn}-m$.
Since $L\subseteq L^{\perp_{\sym}}$,
\[
    \Pr{}{z\in L^{\perp_{\sym}}\setminus L}
    ~=~
    \frac{p^{n_{\inn}-m}(p^{2m}-1)}{p^{2n_{\inn}}-1}.\qedhere
\]
\end{proof}

Note that using the above~\cref{lem:one_inner_coset} and taking a union bound over all possible outercodewords, one obtains an analog of~\cref{prop:tentative_fixed_vector}. This can be used to mimic the proof of~\cref{cor:degen_cor} to obtain a much weaker bound for all $z$: 
\begin{align}
\nonumber\Pr{\Phi, L_\inn}{z\in
    L_{\concat}^{\perp_{\sym}}\setminus L_{\concat}}
    ~&\approx~
  (p)^{-(1+r)n_\inn|S|}  \cdot A_S(L_{\out}^{\perp_{\sym}}\setminus L_{\out})\,\\
\label{eq:pointwise_rater}   ~&\approx~ p^{-n_\inn\cdot n \cdot \parens[\big]{\frac{(1+r)(1-R)}{2} -\gamma}} \gg p^{-n_\inn\cdot n \cdot \parens[\big]{1-rR -\gamma}}.
\end{align}
This bound is quite weak and, in particular, it is insufficient to prove the GV bound. To remedy this, we will take the support size, $|S|$, of $z$ into account. While this will not give a generic bound as before, it will suffice to prove the GV bound.  We will use the following definition to simplify our computations. 

\begin{definition}[Entropy and Capped Entropy] Recall that the $q$-ary entropy function is
\[
    H_q(t)
    ~:=~
    t\log_q(q-1)-t\log_qt-(1-t)\log_q(1-t), \quad t \in (0,1)\;.
\]
with $0\log_q0:=0$. We will use the following ``capped version'' to help with Hamming ball estimates:  
	\[
	\capH_{q}(t) ~=~ \begin{cases}
		H_{q}(t) \qquad &t \leq 1-1/q,\\
		1 \qquad &t \geq 1-1/q\;.
	\end{cases}
	\]
\end{definition}

We are now ready to bound the probability that any fixed vector of weight $w$ is mapped to a low-weight vector after concatenation. 

\begin{lemma}[Single vector bound]\label{lem:weight_w}
Fix $y\in L_{\out}^{\perp_{\sym}}\setminus L_{\out}$ with
$\wt_q(y)=w$. Then, for every integer $h\geq0$,
\[
    \Pr{\Phi, L_\inn}{\norm{\Phi(y)}_{L_\inn}~\leq~h}
    ~\leq~
     \left(\frac{p^{-(n_{\inn}+m)}}{1-p^{-2n_{\inn}}}\right)^w
    (p^2)^{wn_{\inn}\capH_{p^2}(h/(wn_{\inn}))}.
\]
\end{lemma}
\begin{proof}
Let $S:=\supp(y)$, so $|S|=w$. By definition,
\[
    \norm{\Phi(y)}_{L_\inn}
    ~=~
    \sum_{i\in S}\norm{\phi_i(y_i)}_{L_i}.
\]
Thus, the event $\norm{\Phi(y)}_{L_\inn}\leq h$ occurs precisely when there are vectors
\[
    \braces[\big]{ z_i\in L_i+\phi_i(y_i)}_{i\in S}\quad \text{ such that } \quad\sum_{i \in S} \mathrm{wt}_p(z_i) \leq h\,.
\]
For any fixed tuple with
$z_i\neq 0$ for every $i\in S$, independence of the inner encodings and
\cref{lem:one_inner_coset} give
\[
    \Pr{}{
        z_i\in L_i+\phi_i(y_i)
        \text{ for every }i\in S
    }
    ~=~
    \left(\frac{p^{n_{\inn}-m}}{p^{2n_{\inn}}-1}\right)^w.
\]
A tuple $(z_i)_{i\in S}$ is a vector of length $wn_{\inn}$ over the Pauli alphabet of size $p^2$. Hence, the number of tuples with total weight
at most $h$ is at most the Hamming-ball volume
\[
    \sum_{j~=~0}^h\binom{wn_{\inn}}{j}(p^2-1)^j
    ~\leq~
    (p^2)^{wn_{\inn}\capH_{p^2}(h/(wn_{\inn}))}.
\]
For $h/(wn_{\inn})\leq1-1/p^2$, this is the standard entropy estimate. Above
that radius, it follows from the trivial bound $(p^2)^{wn_{\inn}}$ and the
capped-entropy convention.
\end{proof}

Given the above computation for a single outer logical vector, we can now take a union bound over all outer logical vectors $y \in L_{\out}^{\perp_{\mathrm{sym}}} \setminus L_{\out}$ to bound the probability that $L_{\con}$ has small distance.
\begin{proposition}[Distance of random concatenated code]
\label{prop:concat_union_bound}
For every integer $h\geq0$,
\begin{align*}
    \Pr{}{d(L_{\concat})~\leq~ h}
    ~&\leq~
    \sum_{w~=~d_{\out}}^n
    A_w
    \left(\frac{p^{-(n_{\inn}+m)}}{1-p^{-2n_{\inn}}}\right)^w
    (p^2)^{wn_{\inn}\capH_{p^2}(h/(wn_{\inn}))}\,.
 \end{align*}

If the outer code is $\eta$-degenerate below $d_0\leq d$, then for every integer $h\geq0$,    
   \[
    \Pr{}{d(L_{\concat})~\leq~ h}
    ~\leq~
    |L|^\eta
    \sum_{w~=~d_0}^n
    \binom{n}{w}(p^{2m})^{w-d_0+1}
    \left(\frac{p^{-(n_{\inn}+m)}}{1-p^{-2n_{\inn}}}\right)^w
    (p^2)^{wn_{\inn}\capH_{p^2}(h/(wn_{\inn}))}.
\]
\end{proposition}

We are now ready to show that our general ensemble approaches the GV bound. More specifically, we show in \Cref{thm:eta_stabilizer_gv2} that any outer code that is not too degenerate, concatenated with random inner codes, has distance approaching the GV bound. The proof follows Ouyang's strategy from \cite{Ouy14}, but bounds the resulting entropy exponent
directly instead of using the interval minimization as in~\cite[Theorem~III.1]{Ouy14}.

\begin{theorem}[Randomly concatenated not-too-degenerate codes approach the GV bound]\label{thm:eta_stabilizer_gv2}
Fix a prime $p$ and rates $R,r\in (0,1)$ satisfying $ R
    <
    1-2H_{p^2}(1-p^{r-1})$
For every $\eps>0$, there are $\eta,\zeta>0$ and an integer $n_0$ such
that the following holds. Let $n_{\inn}\geq n_0$ be an integer, and let $m=rn_{\inn}$. Let $L_\out$ define a family of $[[n,Rn,d]]_{p^m}$ outer codes that is $\eta$-degenerate below $d_0$ for $d_0~\geq~(1-\zeta)\frac{1-R}{2}n$.
Then, with probability $1-\exp(-\Omega_\eps(n))$ over the independent random inner codes, the concatenated code has rate $rR$ and relative
distance at least $ \deltagv(rR)-\eps$.
\end{theorem}

\begin{proof}
We begin the proof by setting some parameters. The assumption on $R$ implies that,
\[
    \capH_{p^2}(1-p^{r-1})
    ~<~
    \frac{1-R}{2}
    ~<~
    \frac{1-rR}{2}
    ~=~
    \capH_{p^2}(\deltagv),
\]
and hence $\delta_{\mathrm{GV}}>1-p^{r-1}$. Replacing $\eps$ by a smaller
positive number only strengthens the conclusion, so assume $0 <\eps <\deltagv-(1-p^{r-1})$. We now  set
\begin{equation}\label{eq:entropy_gap}
\delta
    ~:=~
    \deltagv-\eps, \quad
    h
    ~:=~
    \left\lfloor\delta nn_{\inn}\right\rfloor,
    \quad
    \kappa_\eps
    ~:=~
    H_{p^2}(\deltagv)-H_{p^2}(\delta)
    ~>~
    0.
\end{equation}
Choose $\zeta,\eta>0$ and $n_0$ so that
\[
    \frac{r(1-R)}{2} \zeta
    ~\leq~
    \frac{\kappa_\eps}{8},
    \qquad
    \frac{\eta r(1-R)}{2}
    ~\leq~
    \frac{\kappa_\eps}{8},
    \qquad
    \frac{2\log_{p^2}2}{n_0}
    ~\leq~
    \frac{\kappa_\eps}{8}.
\]

\medskip
Applying~\cref{lem:weight_w} and taking a union bound over all choices of $y$, we get: 
\begin{align}
\nonumber    \Pr{}{d(L_{\concat})~\leq~ h}
    ~&\leq~
    \sum_{w~=~d_{\out}}^n
    A_w(L_\out^\perpsym\setminus L_\out)
    \left(\frac{p^{-(n_{\inn}+m)}}{1-p^{-2n_{\inn}}}\right)^w
    (p^2)^{wn_{\inn}H_{p^2}(h/(wn_{\inn}))}\\
\nonumber   ~&\leq~   
    \sum_{w~=~d_0}^n
    |L_\out|^{\eta}\binom{n}{w}(p^{2m})^{w-d_0+1}
    \left(\frac{p^{-(n_{\inn}+m)}}{1-p^{-2n_{\inn}}}\right)^w
    (p^2)^{wn_{\inn}\capH_{p^2}(h/(wn_{\inn}))}\;,\\
\label{eq:gv_weight}    ~&=:~   
    \sum_{w~=~d_0}^n T_w \;.
 \end{align}
The second inequality uses the $\eta$-degeneracy assumption of the outer code, and applies~\cref{cor:degen_cor}. Let $T_w$ be the $w^{\mathrm{th}}$-summand above.  We set $x:=w/n$, and the key claim is that:
\begin{equation}\label{eq:Tw_wts}
T_w ~\leq~ (p^2)^{-\kappa_\eps \cdot n \cdot n_{\inn} / 2}.
\end{equation}
To see this, first note that
\begin{align*}
T_w
~&\leq~ 
|L|^{\eta}2^n(p^{2m})^{w-d_0+1} 2^w \left( p^{-(n_{\inn} + m)}\right)^w (p^2)^{w n_{\inn} \capH_{p^2}(h/(wn_{\inn}))} 
\end{align*}
where above we have bounded ${n \choose w} \leq 2^n$ and $(1 - p^{-2n_{\inn}})^{-w} \leq 2^w$.  Then, taking the logarithm base $p^2$ and dividing by $n\cdot n_{\inn}$, we have
\begin{align}
\frac{\log_{p^2}T_w}{n\cdot n_{\inn}} ~&\leq~ 
\eta \frac{ \log_{p^2}|L|}{n \cdot n_{\inn}} + \frac{n+w}{n \cdot n_{\inn}} \log_{p^2}(2) + \frac{(w - d_0 +1)m}{n\cdot n_\inn} - \frac{w(n_{\inn} + m)}{2 n n_{\inn}} + \frac{w n_{\inn} \capH_{p^2}(h/w n_{\inn})}{n \cdot n_{\inn}} \notag \\
~&=~ \frac{\eta (1-R)r}{2} + \frac{ 1 + x }{n_{\inn}} \log_{p^2}(2) + \left( rx - \frac{ r(d_0 - 1) }{n} \right) - \frac{ x(1 + r) }{2} + x \capH_{p^2}\left( \frac{h}{w n_{\inn}} \right) \notag \\
~&=~ \frac{\eta (1-R)r}{2} + \frac{ 1 + x }{n_{\inn}} \log_{p^2}(2)   - \frac{ r(d_0 - 1) }{n} - \frac{ x(1 - r) }{2} + x \capH_{p^2}\left( \frac{h}{w n_{\inn}} \right) \notag \\
~&\leq~ \frac{\eta (1-R)r}{2} + \frac{ 2  }{n_{\inn}} \log_{p^2}(2)   - \frac{rd_0}{n} + \frac{r}{n}  - \frac{ x(1 - r) }{2} + x \capH_{p^2}\left( \frac{h}{w n_{\inn}} \right).\label{eq:finite_stabilizer_exponent}
\end{align}
Above, in the first equality, we have used the facts that $x = w/n$ and $r = m/n_{\inn}$.  We have also used the fact that $|L| = q^{n(1-R)} = p^{mn(1-R)}$, as $L$ defines the outer code of dimension $q^{Rn}$.

At this point, our goal is to show that \Cref{eq:finite_stabilizer_exponent} is at most $-\kappa_\eps / 2$, which will establish \Cref{eq:Tw_wts}.  To that end, we bound \Cref{eq:finite_stabilizer_exponent} by
\[ \frac{ \log_{p^2} T_w }{n n_{\inn}} ~\leq~ 
\underbrace{
\left( 
 x \capH_{p^2}\left( \delta_*/x \right)
- \frac{ x(1 - r) }{2} - \frac{r(1-R)}{2} 
\right)}_{\text{Term I}}
~+~ 
\underbrace{
\left( 
 \frac{\zeta (1-R)r}{2} + 
\frac{\eta (1-R)r}{2} 
+ \frac{ 2 \log_{p^2}(2) }{n_{\inn}}   + \frac{r}{n}  
\right)}_{\text{Term II}},
\]
where above we have used the assumption from the theorem statement that $\frac{d_0}{n} \geq \frac{(1 - \zeta)(1-R)}{2}$; and the fact from the definition of $h$ that $\frac{h}{wn_\inn} \leq \frac{\delta}{x}$ and the fact that the capped entropy function, $\capH(\cdot)$, is increasing. We bound each of these terms below, starting with the second.

\medskip\noindent 
\emph{Term II.}
By our assumptions earlier, each of the first three summands in Term II are at most $\frac{\kappa_\eps}{8}$.  The final term tends to zero as $n$ grows, so for sufficiently large $n$, we see that Term II is less than $\kappa_\eps / 2$.

\medskip
\noindent\emph{Term I.}
It remains to bound Term I by $-\kappa_\eps$; this will establish that \Cref{eq:finite_stabilizer_exponent} is at most $-\kappa_\eps / 2$, as desired.  To do this, we take the derivative of this term with respect to $x$. On the uncapped branch, we have
\[
    \frac{d}{dx}\left[
        x\capH_{p^2}(\delta/x)-\frac{1-r}{2}x-\frac{r(1-R)}{2} 
    \right]
    ~=~
    -\log_{p^2}(1-\delta/x)-\frac{1-r}{2}.
\]
Since $\delta/x\geq\delta>1-p^{r-1}$, we have
$-\log_{p^2}(1-\delta/x)>(1-r)/2$, so this derivative is positive.  
On the
capped branch, $\capH_{p^2}$ is equal to $1$, so the derivative is just $1 - (1-r)/2 = (1+r)/2 > 0$.
In both cases, the derivative is positive, so we see that Term I is increasing for $x \in [d_0/n, 1]$; hence it is maximized at $x = 1$.
Thus, we have
\begin{align*}
(\text{Term I}) ~&\leq~ \capH_{p^2}(\delta) - \frac{1-r}{2} - \frac{r(1-R)}{2}  \\
~&=~ \capH_{p^2}(\delta) - \capH_{p^2}(\deltagv) 
~=~ -\kappa_\eps.
\end{align*}
Together, the bounds on Term I and Term II establish \cref{eq:Tw_wts}, \ie
\[
    T_w
    ~\leq~
    (p^2)^{-\kappa_\eps nn_{\inn}/2}, \quad \text{for every } w\in[d_0,n]  \;
.\]
Summing over at most $n$ possible weights and applying \cref{eq:gv_weight} gives:
\[
    \Pr{}{d(L_{\concat})~\leq~h}
    ~\leq~
    n(p^2)^{-\kappa_\eps nn_{\inn}/2}
    ~=~
    \exp(-\Omega_\eps(n)).
\]
Outside this event, 
we have, 
$ d(L_{\concat}) ~\geq~ h+1 ~\geq~ (\deltagv(rR) -\eps)\cdot  n n_{\inn}. $
Finally, \Cref{clm:concat_perp} implies that $L_{\concat}$ has rate $mRn/(n\cdot n_{\inn}) = rR$, and thus, defines a $[[N, rRN, (\deltagv(rR) -\eps)N]]_p$ code.
\end{proof}

\subsection{Decoding}

Both the random models we have discussed so far yield QLDPC codes approaching the GV bound. While our first random model with trivial inner codes is simpler to describe and analyze, the second one with general inner codes has additional features. For example, codes in the second model have a worst case lower bound on their distance in the form of the product of distances of the outer and inner codes. If both of these values are constants bounded away from 0, the concatenated code also has a constant distance bounded away from 0. While this product bound will always fall short of the GV bound, it holds true for every code in the ensemble, and can be certified simply by checking the distances of the inner codes used in concatenation. We now verify this claim.

\begin{lemma}[Worst-case distance for concatenated codes]\label{lem:prod_distance}
	If the stabilizer code $L_{\out}$ has distance at least $\delta_{\out} n$ and for each $i\in [n]$, the code $L_i$ has distance at least $\delta_{\inn}n_{\inn}$, then the code $L_{\concat}$ has distance at least $\delta_{\out} \delta_{\inn} n n_{\inn}$.
\end{lemma}
\begin{proof}
	Recall (\cref{clm:concat_perp}) that the distance of the concatenated code is
    \[
    	d
    ~=~
    \min_{y\in L_{\out}^{\perp_{\sym}}\setminus L_{\out}}
    \norm{\Phi(y)}_{L_{\inn}}  ~=~ \min_{y\in L_{\out}^{\perp_{\sym}}\setminus L_{\out}}\;\;
\sum_{i=1}^n \norm{\phi(y_i)}_{L_{i}}   \;.
    \]
    For any $y \in L_{\out}^{\perp_{\sym}} \setminus L_{\out}$, let $S = \{i \mid y_i \neq 0\}$. Then, $|S| \geq \delta_{\out} n$ by the distance of the outer code. Moreover, for each $i \in S$, we have that $\phi(y_i) \in W_i \setminus\{0\}$. Since $L_i^{\perp_{\sym}} = L_i \oplus W_i$, we must have 
    \[
    \phi(y_i) \in L_i^{\perp_{\sym}} \setminus L_i ~\implies~ \norm{\phi(y_i)}_{L_i} \geq \delta_{\inn} n_{\inn} ~\implies~ \norm{\Phi(y)}_{L_{\inn}} \geq \delta_{\out} n \cdot \delta_{\inn} n_{\inn} \mper \qedhere
    \]
\end{proof}

As a concrete use-case for such a certifiable lower bound on the distance, we show that the codes in our ensemble can correct a constant fraction of errors -- a feature missing from usual random stabilizer codes.
\begin{lemma}
	Let $L_{\out}$ and $\{L_i\}_{i \in [n]}$ have distances as in~\cref{lem:prod_distance}, and further, let $L_\out$ admit a decoder for errors of weight up to $\theta_\out n_\out$ that runs in time $T_\out(n)$. Then, there exists a decoder for $L_\con$ that runs in time $T_\out(n) + O_{n_\inn}(n)$, and corrects errors upto weight $\theta_{\out}\cdot \frac{\delta_{\inn}}{2}\cdot n n_{\inn}$.
\end{lemma}
\begin{proof}
	For any coordinate $i\in [n]$, suppose $w,w'\in L_i^{\perp_{\sym}}$ are such that $w+L_i$ and $w'+L_i$ are distinct. Then because of the distance of the inner code, it must hold that any element of $w+L_i$ is at least $\delta_{\inn}n_{\inn}$ distance away from any element of $w'+L_i$. This means that there is at most one coset $w+L_i$ that can be at distance less than $\frac{\delta_{\inn}}{2}n_{\inn}$ from any string in $\F_p^{2n_{\inn}}$.
	
	Now, given any corrupted codeword $g \in \F_p^{2n_{\inn}n}$, the decoder splits it into $n$ blocks corresponding to the coordinates of the outer code. Suppose there is a (unique) coset $w_i+L_i$ that is at distance less than $\frac{\delta_{\inn}}{2}n_{\inn}$ away from $g_i$ for some $i\in [n]$. The decoder will find such a coset $w_i+L_i$, and choose $w_i$ to be in $W_i$. If some coordinate does not have a coset close to $g_i$, then $w_i$ is an arbitrary element of $W_i$. Finally, the decoder constructs a new string $(\phi_1^{-1}(w_1),\ldots,\phi_n^{-1}(w_n))$, and feeds it to the outer code's decoder. The output of this outer code's decoder is passed through the map $\Phi$, and 
	the coset of $L_{\concat}$ containing the image of this map is the output of the concatenated code's decoder.
	
	To see why this decoder can correct $\frac{\theta_{\out} \delta_{\inn}}{2}\cdot n n_{\inn}$ errors, let $h\in L_{\concat}^{\perp_{\sym}}$ be a codeword close to $g$. Because of degeneracy, we can only hope to recover the coset $h+L_{\concat}$ and not $h$ itself.
	
	Note that the only blocks for which the decoder did not choose the coset $w_i+L_i$ to be the one containing $h_i$ is when the distance between $h_i$ and $g_i$ is at least $\frac{\delta_{\inn}}{2}n_{\inn}$. Because of the total distance between $g$ and $h$ being bounded, the number of such blocks is at most $\theta_{\out}n$, and therefore the string fed to the outer code decoder is only $\theta_{\out}n$ distance away from $\Phi^{-1}(h')$, where $h' \in h+L_{\inn}$ is such that $h' \in \oplus_{i=1}^n W_i$. The outer code's decoder will return the coset $\Phi^{-1}(h')+L_{\out}$, which can easily be used to recover $h+L_{\concat}$ because
	\begin{align*}
		\Phi(\Phi^{-1}(h')+L_{\out}) + L_{\concat} &= \Phi(\Phi^{-1}(h'))+\Phi(L_{\out}) +L_{\concat}\\
		&= h'+\Phi(L_{\out}) + L_{\inn}+\Phi(L_{\out}) \\
		&= h + L_{\inn} + \Phi(L_{\out}) \\
		&= h+ L_{\concat} \qedhere
	\end{align*} 
\end{proof}
Note that even the trivial inner codes used in random model 1 have distance $\delta_{\inn} = \frac{1}{n_{\inn}}$, but using larger distance inner codes as in random model 2 allows us to design codes with much larger decoding radii. For example, if $p=2$, then $n_{\inn}$ will have to grow as codes approaching Singleton bound need to have large alphabet, and this will disallow a constant (relative) decoding radius.

\subsubsection{Purity Testing and Authentication}

Motivated by problems in quantum message authentication, Barnum \etal~\cite{BCGST02} defined the notion of an $\eps$-\textit{purity testing code}, and used them for ``purity testing'', \ie checking if a pair of quantum states is a perfect Einstein-Podolsky-Rosen (EPR) pair.

\begin{definition}[Purity Testing Codes~\cite{BCGST02}]
	An ensemble of stabilizer codes $\{L_\Phi\}_\Phi \leq \F_p^n$ is an $\eps$-\textit{purity testing code} if for every $z \in \F_p^n$, we have that, $\Pr{\Phi}{z \in L_\Phi^\perpsym\setminus L_\Phi} \leq \eps$.
\end{definition}

The construction of such codes due to~\cite{BCGST02} uses rational curves in projective space. This was simplified in the more recent work of Tr\'enyi, \etal ~\cite{TBGC26} that uses classical error-correcting codes instead of curves to define such codes. However, neither of these constructions is LDPC, and it is not known if they admit efficient decoding algorithms.  As a direct consequence of~\cref{eq:pointwise_rater}, we obtain a new family of purity testing codes that are both QLDPC and decodable.

\begin{corollary}[QLDPC purity testing codes]\label{cor:tentative_purity_testing}
For every $\gamma >0$, the ensemble of QLDPC codes from \cref{thm:final_rate1}, is a purity-testing family with error
$
    \epsilon
    ~\leq~
    p^{-N\parens[\big]{\frac{(1-R)(1+r)}{2}-\gamma}}.
$
Every member of the family is $O_\gamma(1)$-QLDPC and admits a polynomial decoding algorithm.
\end{corollary}

\section{Low-Degeneracy QLDPC Outer Codes from AEL}\label{sec:AEL}
The goal of this section is to construct $\eta$-degenerate codes of any desired degeneracy parameter $\eta$, rate $R$, and distance approaching the quantum Singleton bound. Before we begin describing our construction, we make an observation that any asymptotically good $\eta$-degenerate code must have checks of weight $\Omega(1/\eta)$.

\begin{observation}[QLDPC codes cannot be too degenerate]
	Let $L$ define a $[[n,Rn, \delta n]]_{p^m}$ code. Assume that $L$ has a basis such that each vector has weight at most $t$. Then, $A_{\delta n}(L) \geq |L|^{\frac{\delta}{m(1-R)t} }$. In particular, if $L$ is $\eta$-degenerate, it must have checks of weight at least $t \geq \frac{\delta}{m(1-R)\eta}$.
\end{observation}
\begin{proof}
The sum of any set of at most $\delta n/t$ vectors from the basis has weight at most $\delta n$. The number of basis vectors is $(1-R)n$. Hence, $A_{\delta n}(L) ~\geq~ {(1-R)n \choose \delta n/t} p^{\delta n/t}$. This implies the bound.

\end{proof}

Thus, one cannot hope to have check weights independent of the degeneracy parameter, $\eta$. We will give a construction that has weights of polynomially large in $1/\eta$, which suffices for our purposes as $\eta$ will be a constant. We will prove the following:

\AELCode

Our construction will use the graph amplification technique of Alon, Edmonds, and Luby~\cite{AEL95}. This technique was first adapted for quantum stabilizer codes in the work of Bergamaschi, Golowich, and Gunn~\cite{BGG24} and later in~\cite{BJMST24}. We quickly define the AEL construction using the linear-algebraic language of stabilizer codes that we have developed.

\subsection{The AEL Code}
The construction of~\cite{BJMST24} utilizes Calderbank-Shor Steane (CSS) codes which are a class of stabilizer codes in which the stabilizers can be written as a direct sum, \ie $L = L_X \oplus L_Z$, where $L_X, L_Z \subseteq \F_q^n$. The AEL-based construction of~\cite{BJMST24} can be seen as a concatenation of two CSS codes, plus a ``shuffle'' that can be effectively seen as a different weight function on $\F_p^{n}$.  To describe this shuffle (or the weight function) we need an \textit{expander graph}, and a key property of this graph, that says that it shuffles vertices pseudorandomly. 

\begin{definition}[Expanders and Mixing Lemma]\label{def:expanders}
	An $(n ,\lambda, \Delta)$- bipartite expander is a $\Delta$-regular bipartite graph $G=(U,V,E)$ with $|U|=|V|=n$ such that the normalized second singular value of its adjacency matrix is at most $\lambda$.	Such graphs satisfy the expander mixing lemma (EML) which states that for any subsets $S\subseteq U, T\subseteq V$, let $E(S,T)$ denote the number of edges between $S$ and $T$. Then,
	\[
	\abs[\Big]{E(S,T) - \frac{\Delta}{n} |S||T| } ~\leq~ \lambda \Delta\cdot \sqrt{|S||T|}\;.
	\]

\end{definition} 

Such graphs exist, and in fact, can be constructed efficiently in polynomial time.

\begin{theorem}[Explicit Expanders~\cite{Alon21}]\label{thm:expander_alon}
	For every $\lambda\in(0,1)$, there is a degree $\Delta \leq  \frac{10}{\lambda^2}$ for which such a construction exists for all sufficiently large $n$. 
\end{theorem}

Now we are ready to state the AEL construction  
The idea is to form a concatenated code with $L_\base$ as a base code and $L_0$ as each inner code; and then ``fold'' the concatenated code according to the bipartite expander $G$. The distance bound from \cite[Thm. 4.2]{BJMST24} states that $L_\ael$ has good distance if $L_0$ and $L_\base$ do, and $G$ is a good-enough expander.

\begin{definition}[AEL construction]\label{def:AEL} Fix a prime $p$, let $q=p^b$ for some integer $b$. Let $G = (U,V,E)$ be a $(n,\Delta, \lambda)$-expander. 
Let $L_\base$ be a $[[n, Rn]]_{q^k}$ CSS code and let $L_0$ be a $[[\Delta, k]]_{q}$ CSS code. Then, let $\tilde{L}_\ael = \Phi(L_\base) \oplus L_0^{\oplus n}$ be the code obtained by concatenation of $L_\base$ with $L_0$ as the inner code in each coordinate.  Since, $\tilde{L}_\ael \subseteq (\F_q^2)^{n\Delta}$, we can view it as a subset of $\F_{q}^2[E] := \{f: E \to \F_q^2\}$. For $v \in V$, let $f_v = (f_e)_{v \in e} \in (\F_q^{2})^{\Delta}$ be the vector obtained by bundling the $\Delta$-edges incident on $v$. Then, we define: 
\[
L_\ael ~=~ \fold{\tilde{L}_\ael} ~=~ \braces[\big]{(f_v)_{v \in V} \mid f \in \tilde{L}_\ael }\;.
\]
From \cite[Thm. 4.2]{BJMST24}, $L_\ael$ is a $[[n, \frac{kR}{\Delta} n, d]]_{q^\Delta}$ stabilizer code
where,
\[ \frac{d(L_\ael)}{n} ~=~ \frac{1}{n}\cdot \min_{f \in \tilde{L}_\ael^\perpsym\setminus \tilde{L}_\ael} |\braces{v \in V \mid f_v \neq 0}| ~~\geq~~  \frac{d(L_0)}{\Delta} - \frac{ \lambda n }{d(L_\base)}\;. \]
We will use the notation $\supp_V(f) := \braces{v \in V \mid f_v \neq 0}$, and similarly, $\supp_U(f)$.
\end{definition}

\paragraph{Low-Weight Stabilizers} Now that we have defined our AEL code, the goal is to understand the low-weight stabilizers, so that we may bound the degeneracy and use it as an outer code for \Cref{thm:eta_stabilizer_gv2}.
We first recall~\cref{clm:concat_perp} and state the explicit characterization of the dual space. Let $L_0^\perpsym = W_0\oplus L_0$, and let $f \in L_\ael^\perpsym$. Then, there exists an $x \in L_\base^\perpsym$  such that for any $u \in U$,
\begin{equation}\label{eq:dual}
	f_u ~=~ \phi_u(x_u) + h_u,\; \text{where} \quad  \phi_u(x_u) \in W_0, \;h_u\in  L_0
\end{equation}
Since $W_0\cap L_0 = \{0\}$, we have $f_u = 0 \iff \phi_u(x_u) = h_u = 0$. In the next lemma, we will prove that if $f$ has small weight, then many $f_u$ must be $0$. Thus, from the above observation, both $x,h$ must be low-weight as well. This will help us prove our weight bounds.

\begin{lemma}[Vectors with low-weight right views have low-weight left views]\label{lem:ael_support}
Suppose $0\leq\tau<\alpha\leq1$.  Let $f \in \F_{q}^E$ so that $|\supp_V(f)|\leq\tau n$, and suppose that every
nonzero left view of $f$ (that is, $f|_{u}$ for all $u \in U$) has weight at least $\alpha\Delta$. Then,
\[
    |\supp_U(f)|
    ~\leq~
    \beta(\tau)\cdot n,
    \qquad
    \beta(\tau)
    ~:=~
    \frac{\lambda^2\tau}{(\alpha-\tau)^2}.
\]
\end{lemma}

\begin{proof}
Put $S=\supp_U(f)$ and $T=\supp_V(f)$. All nonzero edges incident to $S$ lie
between $S$ and $T$, so by the expander mixing lemma~\cref{def:expanders}, 
\[
    \alpha\Delta|S|
    ~\leq~
    |E(S,T)|
    ~\leq~
    \tau\Delta|S|+\lambda\Delta\sqrt{\tau n|S|}.
\]
If $S=\varnothing$, the conclusion is immediate. Otherwise, rearranging and
dividing by $\Delta\sqrt{|S|}$ gives
\[
    (\alpha-\tau)\sqrt{|S|}
    ~\leq~
    \lambda\sqrt{\tau n},
\]
which implies the statement.
\end{proof}

Next, we apply \Cref{lem:ael_support} to our setting to bound the number of low-weight stabilizers.

\begin{proposition}[Bounding low-weight stabilizers]\label{prop:ael_dual_enum}
Let ${L_\ael}$ be as in \Cref{def:AEL}, and assume that $L_0$ is pure, \ie every non-zero vector in $h \in L_0^\perp$, $\wt_q(h) \geq \alpha \Delta$, for some $\alpha \in (0,1)$. Fix $0\leq\tau<\alpha$, and set $s=\lfloor\beta(\tau)n\rfloor=  \lfloor \frac{ \lambda^2 \tau }{(\alpha - \tau)^2} n \rfloor$. Then,
\[   A_{\leq \tau n}(L_\ael)
    ~\leq~
 {n \choose s} \cdot q^{ks} \cdot |L_0|^s ~=~ {n\choose s}q^{\Delta s}\;.
\]

\end{proposition}

\begin{proof}

Fix $f \in L_\ael$ with  $|\supp_V(f)| = w \leq \lfloor \tau n \rfloor$.
By \Cref{eq:dual}, 
we may write $f$ uniquely as 
\[
f=\Phi(x)+h, \quad x\in L_\base, \;h\in L_0^{\oplus n}\;. \]
Since every left view of $f = \Phi(x) + h$ lies in $L_0^\perp$, and since the distance of $L_0^\perp$ is at least $\alpha \Delta$ by assumption, every nonzero left view $f|_{\Gamma(u)}$ for $u \in U$ of $f$ has weight at least $\alpha \Delta$. Now we apply \Cref{lem:ael_support} to conclude that $|\supp_U(f)|\leq s$. Now, recalling that $f = \Phi(x) + h$, we observe that
\[ |\supp_U(f)| = \abs[\big]{\supp(x) \;\cup\; \{u \in U\,:\, h|_{\Gamma(a)} \neq 0 \}} ~\leq~ s \;.
  \]
Indeed, this follows since the sum is direct. Therefore, to count the number of such vectors, we may pick a set $S\subseteq[n]$ of size $s$ containing $\supp_U(f)$. Since the base stabilizers supported in $S$ form an isotropic subspace, there are at most $q^{ks}$ choices for $x$. There are at most $|L_0|^s$ choices for $h$. Summing over the ${n\choose s}$ choices of $S$ gives the bound.
\end{proof}

\subsection{Parameter Selection and Locality}

\AELCode*

\begin{proof}
We will construct an AEL code with the required properties.  To do that, we will construct the underlying expander graph, base code, and local code. Choose a rate $R_\base>0$ such that,
\begin{equation}\label{eq:ael_base_defect_choice}
    1- R_\base
    ~\leq~
    \min\left\{
        \frac{\eta(1-R)}{3},
        \frac{\eps(1-R)}{4}
    \right\}.
\end{equation}

First, consider an infinite family of $p$-ary QLDPC codes $\tilde{L}_\base$ with rate at least $R_\base$ and relative distance at least $\delta_\base>0$.   Such a code family exists by~\cite{PK22,LZ23}, and because of Singleton bound and our choice of $R_\base$, this distance $\delta_\base$ must satisfy 
\[
	\delta_\base \leq \frac{\eta(1-R)}{3}
\]

Next, we choose the bipartite expander $G = (A,B,E)$.
Choose a sufficiently large constant $\Delta$ such that
$k:=\Delta R/R_\base$ is an integer, $\Delta-k$ is even, and there is a
$\Delta$-regular bipartite expander family with normalized second singular
value $\lambda\leq10/\sqrt{\Delta}$ satisfying
\begin{equation}
\label{eq:lambda}
    \frac{\lambda}{\delta_{\base}}
    ~\leq~
    \frac{10}{\delta_{\base}\sqrt{\Delta}}
    ~\leq~
    \frac{\eps(1-R)}{4}.
\end{equation}
We can pick a large enough degree $\Delta$  and construct such an expander using~\cref{thm:expander_alon}. Let $G = (A,B,E)$ be a graph from this family, so that $|A| = |B| = n$. 

Next we choose the base code $L_\base$. To do so, first we pick a large enough integer $b$ such that:
\[
    b~\geq~    \frac{3}{2\eta}\cdot \Delta \log_2(p)(1-R).
\]

Set $q_0=p^b$. We pick a $p$-ary QLDPC code $\tilde{L}_\base$ described from the family above of blocklength $nbk$. Then we fold each block of $bk$ coordinates into one symbol of $\F_p^{bk}\cong\F_{q_0}^k$. This gives  us a $[[n, R_\base, {\delta}_{\out}]]_{q_0^k}$ stabilizer code , $L_\base$. Finally, $L_\base$ is still a QLDPC code if $bk$ is constant.

Finally, we pick the 
local code $L_0$ to be a $[[\Delta, k, \alpha \Delta]]_{q_0}$ quantum
Reed--Solomon CSS code. Here, we set the rate $R_0 = k/\Delta=R/R_\base$, and the normalized distance,
\begin{equation}
\label{eq:alpha}	    \alpha
      ~\geq~ \frac{1}{2}\left(1-R_0\right).
\end{equation}

Moreover, this local code is pure, and has the above weight bound for every non-zero vector in $L_0^\perp$, not just in $L_0^\perp \setminus L_0$. From~\cref{def:AEL}, we conclude that the resulting AEL code $L_\ael$ is a $ [[ n, Rn, \delta_\ael n]]_q $ stabilizer code, where:
\begin{align*}
\delta_{\mathrm{AEL}}  ~&:=~ \alpha - \frac{ \lambda }{ {\delta}_{\base} }\;,&&(\text{\cref{def:AEL}})\\
\delta_{\mathrm{AEL}} ~&\geq~ \frac{1}{2} \left( 1 - \frac{R}{R_\base} \right) - \frac{ \eps (1 - R) }{4}\;,  &&(\text{\cref{eq:alpha,eq:lambda}})\\
~&\geq~ \frac{1-R}{2} - (1-R_\base) - \frac{\eps(1-R)}{4} ~\geq~(1 - \eps) \left( \frac{1-R}{2} \right)\;. &&(\text{\cref{eq:ael_base_defect_choice}})
\end{align*}

Set $d_0=\lceil\delta_{\mathrm{AEL}}n\rceil$ and $\tau_n=(d_0-1)/n<\delta_{\mathrm{AEL}}$. Recall that
\[ \beta(\tau) = \frac{ \lambda^2 \tau }{(\alpha - \tau)^2}. \]
Note that since by definition we have $\delta_{\mathrm{AEL}} = \alpha - \lambda/\delta_{\base}$, we have
\[ \beta(\delta_{\mathrm{AEL}}) = \left( \frac{\lambda}{\alpha - \delta_{\mathrm{AEL}}} \right)^2 \delta_{\mathrm{AEL}} = \delta_{\base}^2 \delta_{\mathrm{AEL}}. \]
Then
since $\beta(\tau)$ is increasing for $\tau<\alpha$, we conclude that
\begin{equation}
\label{eq:beta}	    \beta(\tau_n)
    ~\leq~
    \beta(\delta_{\mathrm{AEL}})
    ~=~
    \delta_{\mathrm{AEL}}\delta_{\base}^2
    ~\leq~
    \delta_{\base}
    ~\leq~
    \frac{\eta(1-R)}{3},
\end{equation}

Set $s=\lfloor\beta(\tau_n)n\rfloor$. Since $\tau_n n=d_0-1$, we can apply \Cref{prop:ael_dual_enum} to get:
\begin{align*}
A_{<d_0}(L_\ael) ~&=~ A_{\leq\tau_n n}(L_\ael) ~\leq~ {n\choose s}q_0^{\Delta s},\\
\implies~\log_{q_0}(A_{<d_0}(L_\ael)) ~&\leq~ n\parens[\Big]{H_2(\beta(\tau_n))\log_{q_0}2+\Delta\beta(\tau_n)}\\
&\leq~\eta n\Delta(1-R) ~=~ \eta\log_{q_0}|L_\ael|\;.
\end{align*}
Here, we used that $|L_0| = q_0^{\Delta-k}$. The final inequality follows from our parameter choices for $q_0$ and from~\cref{eq:beta}. Consequently,
$A_{<d_0}(L_\ael)~\leq~|L|^\eta$, and thus, $L_\ael$ is an $\eta$-degenerate code.

From \Cref{def:AEL}, since $b,k,\Delta$ are constants, the AEL code $L_\ael$ is also QLDPC. 
\end{proof}

Finally, we put the pieces together to prove the existence of QLDPC codes over constant-sized alphabets approaching the stabilizer GV bound.
\begin{theorem}[QLDPC codes approaching the stabilizer GV bound]
\label{thm:qldpc_stabilizer_gv}
Fix a prime $p$, a rational rate $\rho\in(0,1)$, and $\eps>0$. There is an
ensemble of $p$-ary QLDPC codes of rate $\rho$ and locality independent of the
blocklength such that, with probability $1-\exp(-\Omega(n))$, a random code in
the ensemble has relative distance greater than
\[
    H_{p^2}^{-1}\left(\frac{1-\rho}{2}\right)-\eps.
\]
\end{theorem}

\begin{proof}
Choose a rational $r<1$ sufficiently close to $1$ that $R:=\rho/r<1$ and
\[
    R
    ~<~
    1-2H_{p^2}(1-p^{r-1});
\]
this is possible because the right-hand side tends to $1$ as $r\to1$. Apply
\cref{thm:eta_stabilizer_gv2} with rates $R,r$, obtaining
$\eta_0,\zeta_0,n_0$, and use \cref{thm:ael_eta_good} with
$\eta_0,\zeta_0$. Choose $m$ so that $n_{\inn}:=m/r\geq n_0$ is an integer, and concatenate
each outer symbol with an independent random $[[n_{\inn},m]]_p$ stabilizer
encoding. The concatenated rate is $rR=\rho$, and the distance claim follows
from \cref{thm:eta_stabilizer_gv}. Locality is preserved because the outer
locality and the inner blocklength are both constant.
\end{proof}

The codes constructed here can be efficiently decoded by leveraging the base code's decoder from \cite{LZ23}, as was shown by \cite{BGG24}. Thus, we could use these AEL codes as outer codes to obtain random QLDPC codes approaching the GV bound, while being efficiently decodable, obtaining quantum analogs of similar results obtained for classical codes~\cite{GI04}.

\section*{Statement of AI use}
ChatGPT Sol 5.6 was used to assist with checking proofs, generating figures, performing literature searches, and drafting parts of this manuscript. The application to analyzing entanglement fidelity for general memoryless channels (\cref{cor:tentative_general_channel}) and the connection to Hamada's paper~\cite{Ham03} were pointed out by ChatGPT. 
All mathematical ideas are due to the authors, and they have written and verified the correctness of all claims and proofs in the manuscript.

\bibliographystyle{alphaurl}
\bibliography{macros,madhur}

\end{document}